\UseRawInputEncoding
\documentclass[%
        reprint,
        onecolumn,
        superscriptaddress,
        nofootinbib,
        aps
    ]{revtex4-2}

\usepackage[english]{babel}

\usepackage{subcaption}
\usepackage{enumerate}

\usepackage[utf8]{inputenc}

\usepackage{graphicx}

\usepackage{hyperref}
\hypersetup{
	colorlinks=true,
	linkcolor=blue,
	citecolor=blue,
	urlcolor=blue,
}

\usepackage{physics}
\usepackage{amsmath}
\usepackage{amsthm}
\usepackage{amssymb}
\usepackage{bm} 
\usepackage{bbm} 
\usepackage{mathtools}

\usepackage{dsfont}
\usepackage{braket}
\usepackage{url}
\usepackage{algorithm}
\usepackage{algpseudocode}

\usepackage{tikz}
\usetikzlibrary{quantikz2}

\usepackage{multirow}

\usepackage{enumitem}
\newlist{nobullet}{itemize}{1}
\setlist[nobullet]{label={}, leftmargin=0pt, labelsep=0pt, itemsep=0pt, topsep=0pt, parsep=0pt}

\usepackage{xspace}
\usepackage{orcidlink}

\newcommand{\R}{\mathbb{R}}

\newcommand{\poly}[1]{\text{poly}\!\left(#1\right)}

\renewcommand{\vec}[1]{\bm{#1}}

\renewcommand{\ket}[1]{|#1\rangle}
\renewcommand{\bra}[1]{\langle#1|}

\newcommand{\aqa}{$\langle aQa^L \rangle$ Applied Quantum Algorithms, Leiden University, The Netherlands}
\newcommand{\lorentz}{Instituut-Lorentz, Leiden University, Niels Bohrweg 2, 2333 CA, The Netherlands}
\newcommand{\liacs}{LIACS, Leiden University, Niels Bohrweg 1, 2333 CA, Leiden, The Netherlands}
\newcommand{\bmw}{BMW Group, 80788 München, Germany}

\makeatother 

\newtheoremstyle{journalrunin}%
  {10pt}{0pt}
  {}{}
  {\itshape}
  {.}{1em}
  {}
\theoremstyle{journalrunin}

\newtheorem{theorem}{Theorem}

\newtheorem{lemma}{Lemma}
\newtheorem{definition}{Definition}
\newtheorem{prop}{Proposition}

\newtheorem{exmp}{Example}
\newtheorem{hyp}{Hypothesis}

\usepackage{xcolor}
\definecolor{color_1}{RGB}{133, 116,  55}
\definecolor{color_2}{RGB}{64,  99,  48}
\definecolor{color_3}{RGB}{32,  73, 116}

\begin{document}
\title{Variational Quantum Generative Modeling by Sampling Expectation Values of Tunable Observables}

\author{Kevin Shen}
\email{kevin.shen@bmwgroup.com}%
\affiliation{\aqa}%
\affiliation{\liacs}%
\affiliation{\bmw}%

\author{Andrii Kurkin}
\affiliation{\aqa}%
\affiliation{\liacs}%
\affiliation{\bmw}%

\author{Adrián Pérez-Salinas}
\affiliation{\aqa}%
\affiliation{\lorentz}%

\author{Elvira Shishenina}
\thanks{Now at Quantinuum, Leopoldstrasse 180, 80804 München, Germany.}
\affiliation{\bmw}%

\author{Vedran Dunjko}
\affiliation{\aqa}%
\affiliation{\liacs}%

\author{Hao Wang}
\email{h.wang@liacs.leidenuniv.nl}%
\affiliation{\aqa}%
\affiliation{\liacs}%


\maketitle

\section{abstract}
Expectation Value Samplers (EVSs) are quantum generative models that can learn high-dimensional continuous distributions by measuring the expectation values of parameterized quantum circuits. However, these models can demand impractical quantum resources for good performance. We investigate how observable choices affect EVS performance and propose an Observable-Tunable Expectation Value Sampler (OT-EVS), which achieves greater expressivity than standard EVS. By restricting the selectable observables, it is possible to use the classical shadows measurement scheme to reduce the sample complexity of our algorithm. In addition, we propose an adversarial training method adapted to the needs of OT-EVS. This training prioritizes classical updates of observables, minimizing the more costly updates of quantum circuit parameters. Numerical experiments, using an original simulation technique for correlated shot noise, confirm our model's expressivity and sample efficiency advantages compared to previous designs. We envision our proposal to encourage the exploration of continuous generative models running with few quantum resources. 

\section{Introduction}

Generative modeling is the task of, given a dataset, learning to generate similar new data. Generative models such as variational autoencoders \cite{kingma_auto-encoding_2022}, diffusion probabilistic models \cite{sohl-dickstein_deep_2015}, and generative adversarial networks (GANs) \cite{goodfellow_generative_2014} have achieved remarkable success in various industrial applications, shaping diverse aspects of our daily lives. However, these models face challenges associated with high computational demands and sustainability concerns \cite{thompsoN_computational_2022, pattersoN_carbon_2021}. As promising alternatives, generative models based on quantum computers have been investigated with some proven advantages for specific artificial problems \cite{gao_efficient_2017, coyle_born_2020, sweke_quantum_2021}. 
In this domain, numerous proposals \cite{tian_recent_2023} have emerged, such as the Quantum Boltzmann Machine \cite{amin_quantum_2018, benedetti_quantum-assisted_2017, kieferova_tomography_2017}, the Quantum Circuit Born Machine \cite{cheng_information_2018, liu_differentiable_2018, benedetti_generative_2019}, and the quantum GANs \cite{lloyd_quantum_2018, dallaire-demers_quantum_2018, zoufal_generative_2021}. In most cases, the models allow the modeling of discrete distributions. 

The Expectation Value Sampler (EVS) \cite{romero_variational_2021, anand_noise_2021, barthe_expressivity_2024} is a comparatively less explored quantum generative model which, in contrast to previous examples, natively models continuous distributions, another example being the quantum diffusion models \cite{zhang_generative_2023, cacioppo_quantum_2023, kolle_quantum_2024}. In EVS, the generated data emerge as expectation values of preselected observables measured on a state sampled from a particular distribution. 
Since the proposal of EVS, there have been some studies exploring its integration with classical neural networks in a hybrid framework as well as benchmarking \cite{riofrio_characterization_2024, kiwit_benchmarking_2024} for applications to image synthesis \cite{huang_experimental_2021, shu_variational_2024, tsang_hybrid_2023, chang_latent_2024} and molecule design \cite{li_quantum_2021, kao_exploring_2023}. 
Yet, many questions regarding the theoretical properties of EVS remain unsolved. In particular, the choice of observables, a key component of EVS, has remained mostly unaddressed. 

In this work, we analyze the role of observables and propose an enhanced EVS model with tunable observables. The choice of tunable observables aims to balance quantum and classical computational resources without sacrificing performance. We provide two results in this direction. First, we show that our tunable-observable models have higher expressivity than fixed-observable models. That is, the set of reachable probability distributions expands. Second, we design a specific observable parameterization that uses classical shadows \cite{huang_predicting_2020} to reduce the sample complexity of our model. The considered observables are Pauli strings with low locality, forming a linear space of dimension super-polynomial in the qubit number. A concurrent work \cite{Shi2024} also studies tunable observables in the EVS model but only addresses the case of $1$-local observables.

We devise a tailored training method for OT-EVS inspired by adversarial training \cite{goodfellow_generative_2014, arjovsky_wasserstein_2017, gulrajani_improved_2017}. The key idea of the training is to update the tunable observables more frequently than the parameters defining the quantum circuit. Since updating the observables requires only classical computation, this can be done without extra quantum measurements. We conduct numerical experiments to benchmark our adapted training method and observe its generally better performance than standard adversarial training for the same amount of quantum resources. 
 
\section{Result}
\label{sec: background}
\subsection*{Expectation Value Sampler}
EVSs are generative models built on parameterized quantum circuits (PQCs) \cite{havlicek_supervised_2019, schuld_quantum_2019, cerezo_variational_2021}, designed for learning continuous distributions. EVSs belong to the class of latent variable models, meaning that data are generated by transforming from some latent random variables. Specifically, latent variables are embedded as the gate parameters of a PQC to prepare a random quantum state. The generated data are the expected values of some observables measured on that state, which form a random vector.  
\begin{definition}[Expectation Value Sampler, adapted from \cite{barthe_expressivity_2024}]
\label{def:evs}
 An Expectation Value Sampler is a generative model specified by a triple of constructive components $(\mathbb{P}_z, U_{\vec{\vec{\theta}}}, (O_m)^{M}_{m=1})$. $\mathbb{P}_z$ is the distribution of some efficiently samplable latent variable $\vec{z}$. $U_{\vec{\vec{\theta}}}$ is some $n$-qubit PQC, parameterized by some $\vec{\vec{\theta}} \in \varTheta \subset \mathbb{R}^d$, that prepares a state $\ket{\psi_{\vec{\vec{\theta}}}(\vec{z})}=U_{\vec{\vec{\theta}}} (\vec{z}) \ket{0}^{\otimes n}$. $(O_m)^{M}_{m=1}$ are some observables defined on the same $n$-qubit system. $U_{\vec{\vec{\theta}}}$ and $(O_m)^{M}_{m=1}$ together build a parametric family of functions $G_{\vec{\theta}}$ that transforms $\mathbb{P}_z$ into the output distribution $\mathbb{P}$. In particular, $G_{\vec{\theta}}$ writes
\begin{equation}
\begin{aligned}
G_{\vec{\theta}}:\mathcal{Z} &\longrightarrow \mathcal{Y} \\
\vec{z} &\longmapsto  \vec{y} = \left( \bra{\psi_{\vec{\vec{\theta}}}(\vec{z})} O_m \ket{\psi_{\vec{\vec{\theta}}}(\vec{z})}\right)^M_{m=1},
\end{aligned}
\end{equation}
where $\mathcal{Z} \subset \mathbb{R}^K$ and $\mathcal{Y} \subset \mathbb{R}^M$ are the supports of $\vec{z}$ and $\vec{y}$ respectively. In sampling mode, the Expectation Value Sampler repeatedly draws $\vec{z} \sim \mathbb{P}_z$ and returns the transformed sample $G_{\vec{\theta}}(\vec{z}) \sim \mathbb{P}$ as output.
\end{definition} 

We consider latent variables that follow a uniform distribution $\mathbb{P}_z = \mathcal{U}([-\pi, \pi]^K)$ throughout this work for simplicity. By fixing $\mathbb{P}_z$, we will then denote an EVS by its function family $G_{\vec{\theta}}$. 

An interesting property of EVSs is that the output dimension $M$ and the number of qubits $n$ are not intrinsically bound to each other but should be separately determined based on the learning task. EVSs of as few as $\Theta(\log(M))$ qubits are proven to be universal approximators for $M$-dimensional distributions \cite{barthe_expressivity_2024}, provided $U_{\vec{\theta}}$ has infinite circuit depth. This finding underscores the capability of EVSs to learn high-dimensional distributions with few qubits. 

\subsection*{Observable-Tunable Expectation Value Sampler}
\label{sec: expressivity and parameterization}
The main contribution of this paper is the Observable-Tunable Expectation Value Sampler (OT-EVS), a generalization of the EVS with fixed observables (Definition~\ref{def:evs}, hereafter termed OF-EVS).

\begin{definition}[Observable-Tunable Expectation Value Sampler (OT-EVS)] \label{def:ot-evs}
An Observable-Tunable Expectation Value Sampler is a generative model specified by a quadruple of constructive components $(\mathbb{P}_z, U_{\vec{\vec{\theta}}}, (O_l)^{L}_{l=1}, \vec\alpha)$. The quantity $\vec\alpha \in \mathbb R^{M \times L}$ is a tunable weight matrix. For every fixed $\vec\alpha$, the generative model is given by the EVS specified through the triple $
(\mathbb{P}_z, U_{\vec{\vec{\theta}}}, (A_m)^{M}_{m=1})$, where $A_m = \sum\limits_{l = 1}^L \vec\alpha_{m, l} O_l$.
\end{definition}

This formulation separates the quantum (EVS) and classical (linear combination) contributions in OT-EVS, enabling precise evaluation of each component's role. In most practical settings, making the observables tunable enhances the model expressivity. We formally defining the notion of \textit{relative expressivity} as: 

\begin{definition}[Relative Expressivity] Let $G_{\vec{\vec{\theta}}}$ and $H_{\vec{\phi}}$ be two parametric families of random variables with finite variances. We define $S_G \coloneq \{G_{\vec{\vec{\theta}}}\colon \mathbb{R}^K \rightarrow\mathbb{R}^M | \vec{\vec{\theta}}\in \mathbb{R}^{d}\}$ and $S_H$ similarly. We say:
\begin{nobullet}
    \item $G_{\vec{\vec{\theta}}}$ is at least as expressive as $H_{\vec{\phi}}$, if $\forall h \in S_H, \; h \in S_G$. We write $H_{\vec{\phi}} \preceq G_{\vec{\vec{\theta}}}$. $\preceq$ is a non-strict order.
    \item $G_{\vec{\vec{\theta}}}$ is strictly more expressive than $H_{\vec{\phi}}$, if $H_{\vec{\phi}} \preceq G_{\vec{\vec{\theta}}}$ and $\exists g \in S_G,\; g \notin S_H$. We write $H_{\vec{\phi}} \prec G_{\vec{\vec{\theta}}}$. $\prec$ is a strict partial order.
    \item $G_{\vec{\vec{\theta}}}$ is as expressive as $H_{\vec{\phi}}$, if $H_{\vec{\phi}} \preceq G_{\vec{\vec{\theta}}}$ and $G_{\vec{\vec{\theta}}} \preceq H_{\vec{\phi}}$. We write $H_{\vec{\phi}} \cong G_{\vec{\vec{\theta}}}$. $\cong$ is the identity relation.
\end{nobullet}
\end{definition}

A few statements then follow naturally. 

\begin{prop}[Expressivity never decreases using tunable observables] \label{prop: never decrease} For any OF-EVS $H_{\vec{\theta}} \coloneq (\mathbb{P}_z, U_{\vec{\vec{\theta}}}, (Q_m)^M_{m=1})$ and any OT-EVS $G_{\vec{\theta}, \vec{\alpha}} \coloneq (\mathbb{P}_z, U_{\vec{\vec{\theta}}}, (O_l)^L_{l=1}, \vec\alpha)$, 
\begin{equation}
    \forall 1\leq m \leq M, \exists 1\leq l \leq L, Q_m = O_l \; \Longrightarrow \; H_{\vec{\phi}} \preceq G_{\vec{\theta}}
    \end{equation}
\end{prop}

\begin{prop}[Expressivity of a universal generative model cannot increase further] \label{prop:no increase universal}
For any OF-EVS $H_{\vec{\theta}} \coloneq (\mathbb{P}_z, U_{\vec{\vec{\theta}}}, (Q_m)^M_{m=1})$ and any OT-EVS $G_{\vec{\theta}, \vec{\alpha}} \coloneq (\mathbb{P}_z, U_{\vec{\vec{\theta}}}, (O_l)^L_{l=1}, \vec\alpha)$, 
\begin{equation}
    G_{\vec{\theta}, \vec{\alpha}} \text{ is universal and } \forall 1\leq m \leq M, \exists 1\leq l \leq L, Q_m = O_l \; \Longrightarrow \; H_{\vec{\phi}} \cong G_{\vec{\theta}}
    \end{equation}
\end{prop}

Despite the existence of the extremal case of universal EVS depicted in Proposition \ref{prop:no increase universal}, OT-EVSs will generally be
strictly more expressive than their analogous OF-EVSs in practice. We provide two characteristic examples below. The verification of these examples is given in the Supplementary Notes.

\begin{exmp}
\label{exmp:toy_exp}
Consider the two-qubit circuit $U_{\vec{\theta}}(\vec{z})=R_Y(z_1 + \theta_1) \otimes R_Y(z_2 + \theta_2)$ with $z_1, z_2 \sim \mathbb{P}_z = \mathcal{U}[- \pi, \pi]$ and $\theta_1, \theta_2 \in \mathbb{R}$. Consider the OF-EVS $(\mathbb{P}_z, U_{\vec{\vec{\theta}}}, (O_1, O_2))$, where $O_1 = X_1 I_2$ and $O_2 = I_1 X_2$, and the OT-EVS $ (\mathbb{P}_z, U_{\vec{\vec{\theta}}}, (O_1, O_2), \vec{\alpha})$, where $\vec{\alpha}\in \mathbb{R}^{2 \times 2}$. Then, the OT-EVS is strictly more expressive than the OF-EVS.
\end{exmp}

\begin{exmp}
\label{exmp:min_rand}
Consider an OF-EVS $(\mathbb{P}_{z}, U_{\vec{\vec{\theta}}}, (O_1))$ that $\forall \vec{\theta} \in \Theta$,
$U_{\vec{\theta}} (\vec z) \ket{0}$ follows the Haar-random distribution over $n$-qubit quantum states, $n>1$. Consider an OT-EVS $(\mathbb{P}_z, U^\prime_{\vec{\vec{\theta}}},(O_1, O_2), \vec\alpha)$, where $\vec{\alpha} \in \mathbb{R}^{1\times2}$, which satisfies the constraint $\forall V \in \mathcal{SU} (2n), V O_1 V\dagger \neq O_2. $ Then, the OT-EVS is strictly more expressive than the OF-EVS.
\end{exmp}

\subsection*{Adapted Adversarial Training} 
We train OT-EVS within the Wasserstein GAN (WGAN) framework \cite{gulrajani_improved_2017, arjovsky_wasserstein_2017}, which aims to minimize the Wasserstein-1 distance between the generated ($\mathbb{P}$) and target ($\mathbb{Q}$) data distributions using the Kantorovich-Rubinstein duality. 
\begin{definition}[Kantorovich-Rubenstein duality~\cite{kantorovich_mathematical_1960}]\label{def:wasserstein_distance} The Wasserstein-1 distance between two probability distributions $\mathbb{P}$ and $\mathbb{Q}$ in $\R^m$ is:
\begin{equation}\label{eq.Wasserstein distance simplified}
    \mathcal{W}_1(\mathbb{P}, \mathbb{Q}) = \sup_{\|D\|_{L} \leq 1} \left( \int_{\mathcal{X}} D(x) \mathrm{d} \mathbb{P} - \int_{\mathcal{Y}} D(y) \mathrm{d}\mathbb{Q} \right),
\end{equation}
where $\mathcal{X}$ and $\mathcal{Y}$ are the supports of $\mathbb{P}$ and $\mathbb{Q}$, and the supremum is taken over the set of $1$-Lipschitz functions $D$. 
\end{definition}

The WGAN framework consists of two competing networks, a generator $G_{\vec{\theta}, \vec{\alpha}}$ that learns to synthesize data and a critic $D_{\vec{w}}$ (the parameterized counterpart of $D$ in Definition \ref{eq.Wasserstein distance simplified}), which maps data samples to real values. The critic is trained to minimize the loss function $\mathcal{L}_C$ in order to approximate $-\mathcal{W}_1(\mathbb{P}, \mathbb{Q})$, as proposed in Ref. \cite{gulrajani_improved_2017}. \begin{equation}
\mathcal{L}_C(\vec{w};\vec{\theta}, \vec{\alpha}) =  \frac{1}{B}\sum_{i=1}^B \left[ D_{\vec{w}}(\vec{y}^{(i)}) - D_{\vec{w}}(\vec{x}^{(i)}) + \lambda \left(\norm{\nabla_{\hat{\vec{x}}}D_{\vec{w}}(\hat{\vec{x}}^{(i)})}_2 -1\right)^2 \right], 
\label{eq:lc}
\end{equation}
where $B$ is the batch size, $\vec y^{(i)} = G_{\vec{\theta} \vec{\alpha}}(\vec{z}^{(i)})$ and $\vec x^{(i)}$ are the generated and training data respectively. The additional gradient penalty term (with $\lambda > 0$, $\hat{\vec{x}}^{(i)} = \varepsilon \vec{x}^{(i)} + (1-\varepsilon) \vec{y}^{(i)}$ and $\varepsilon \sim \mathcal{U}[0,1]$) softly enforces the Lipschitz constant required by the Kantorovich-Rubinstein duality, penalizing deviations of the critic's gradient norm from unity. While the critic could employ a PQC architecture akin to quantum discriminators in QGANs~\cite{dallaire-demers_quantum_2018, hu_quantum_2019, romero_variational_2021, huang_experimental_2021}, we use classical neural networks exclusively for the critic to focus on the aspects of the OT-EVS. The generator is trained to minimize the loss function $\mathcal{L}_G$,
\begin{equation}
     \mathcal{L}_G(\vec{\theta}, \vec{\alpha} ; \vec{w}) = - \frac{1}{B}\sum_{i=1}^B  D_{\vec{w}}(\vec{y}^{(i)}),
     \label{eq:lg}
\end{equation} 
a simplified version of $\mathcal{W}_1(\mathbb{P}, \mathbb{Q})$ that only includes the term dependent on the generator parameters. This objective drives the generated distribution toward the target distribution. 

The two networks are trained alternately, one fixed while the other is updated. In the default setting, we treat $\vec{\alpha}$ and $\vec{\theta}$ as a single set of parameters and update them jointly. However, motivated by resource efficiency, we propose two modified training schemes in which $\vec{\alpha}$ is updated more frequently than $\vec{\theta}$. This approach reduces the consumption of quantum resources (which are expensive in terms of time and energy) while compensating with increased classical resource usage. Specifically, each update of the quantum parameters $\vec{\theta}$ requires $2N_dN_sB$ quantum circuit executions due to the parameter shift rule \cite{mitarai_quantum_2018, schuld_evaluating_2019}, where $N_d$ is the dimension of $\vec{\theta}$ and $N_s$ is the number of measurements per sample. In contrast, updating $\vec{\alpha}$ and $\vec{w}$ only involves $N_sB$ measurements per forward pass, making it significantly less resource-intensive.
We summarize the key differences between the three methods below, while the complete pseudocodes are provided in the Supplementary Notes.

\begin{nobullet}
    \item \textit{Joint:} $\vec{\alpha}$ and $\vec{\theta}$ are updated simultaneously. Each training iteration performs $N_{\vec{w}} \geq 1$ updates to $\vec{w}$, followed by a single joint update to $\vec{\alpha}$ and $\vec{\theta}$.
    
    \item \textit{Asynchronous:} $\vec{\alpha}$ is updated more frequently than $\vec{\theta}$. As in \textit{Joint}, each iteration first performs $N_{\vec{w}}$ updates to $\vec{w}$. Then, $\vec{\alpha}$ is updated $N_{\vec{\alpha}} > 1$ times, followed by one update to $\vec{\theta}$.
    
    \item \textit{Decoupled:} $\vec{\alpha}$ is decoupled with $\vec{\theta}$ but updated jointly with $\vec{w}$. Each iteration consists of $N_{\vec{\alpha}}$ loops, where each loop performs $\lceil N_{\vec{w}} / N_{\vec{\alpha}} \rceil$ updates to $\vec{w}$ followed by one update to $\vec{\alpha}$. After all loops, $\vec{\theta}$ is updated once.
\end{nobullet}

\subsection*{Shadow-Frugal Parameterization of Observables}\label{sec:shadow-frugal}
OT-EVS requires estimating the expectation values with sufficient accuracy in both the training and sampling modes. To minimize sample complexity, we propose restricting $(O_l)^{L}_{l=1}$ to maximally $k$-local Pauli string with $k \in \mathcal{O}(\operatorname{polylog}(n))$. This restriction allows efficient parallel estimation of expectation values using classical shadows. We refer to this strategy as the shadow-frugal parameterization of observables. 
\begin{definition}
    [Shadow-Frugal Parameterization] An OT-EVS $(\mathbb{P}_z, U_{\vec{\vec{\theta}}}, (O_l)^{L}_{l=1}, \vec\alpha)$ employs shadow-frugal parameterization if all observables $O_l$ are $k$-local Pauli strings with $k \in \mathcal{O}(\operatorname{polylog}(n))$.  
 \label{def:shadow frugal}
\end{definition}

To align with the training objective, we define the sample complexity as the number of measurements needed to ensure $\mathcal{W}_1(\mathbb{P}^B, \tilde{\mathbb{P}}^B) < \epsilon$ with high probability, where $\mathbb{P}^B$ and $\tilde{\mathbb{P}}^B$ denote the empirical distributions of a batch of generated samples with and without measurement errors, respectively. 

\begin{theorem}[Sample Complexity]\label{thm:sample_complexity}  
Let $L$ be the number of $k$-local Pauli strings. Let $B$ be the batch size. Denote by $\mathbb{P}^B$ the empirical noiseless distribution of a batch of samples generated from an OT-EVS with shadow-frugal parameterization, and by $\tilde{\mathbb{P}}^B$ the distribution with measurement error. We assume $\norm{\vec\alpha}_{\infty} \leq T$, implying $\norm{A_m}_{\infty}\leq T$, $\forall 1\leq m \leq M$, where $A_m = \sum_{l = 1}^L \vec\alpha_{m, l} O_l$. For $\epsilon, \delta >0$, the following probabilistic bound
\begin{equation}
\operatorname{Pr} \left(\mathcal{W}_1(\mathbb{P}^B, \tilde{\mathbb{P}}^B) \leq \epsilon \right) > 1- \delta
\end{equation}
holds with measurements $N_s$,
\begin{align}
    N_s &\geq 68 \cdot3^k \left\lceil B \frac{T^2}{\epsilon^2} \log \left( \frac{2LB}{\delta} \right) \right\rceil \text{ (with classical shadows)} \label{shadow equation} \\
    N_s &\geq 2 L \left\lceil B \frac{T^2}{\epsilon^2} \log \left( \frac{2LB}{\delta} \right) \right\rceil  \label{conventional equation} \text{(with the conventional measurements)}
\end{align}
\end{theorem}
\begin{proof}
    See the Supplementary Notes.
\end{proof}
The total number of $k$-local Pauli strings with $k \in \mathcal{O}(\operatorname{polylog}(n))$ scales superpolynomially as
\begin{equation}  
L = \sum_{j=0}^{k}{n\choose j}3^{j} \in 2^{\tilde{\Theta} (\log (n)^c)}
\label{eq: L}
\end{equation} with $\Tilde{\Theta}(f(n)) = \Theta(f(n) \cdot \log^k n)$). Theorem \ref{thm:sample_complexity} then implies that a total of $\mathcal{O}(\poly{n}/\epsilon^2)$ measurements is sufficient to estimate all expectation values accurately using classical shadows (equation (\ref{eq: L}) $\rightarrow$ equation (\ref{shadow equation})), which yields a super-polynomial resource advantage over the conventional measurement scheme (equation (\ref{eq: L}) $\rightarrow$ equation (\ref{conventional equation})). 

\subsection*{Numerical Experiments}
We empirically evaluate OT-EVS in the shadow-frugal parameterization (see ``Methods'' and Supplementary Notes for experimental details). All experiments use the same critic architecture for consistency: a three-layer multi-layer perceptron (MLP) with $512$ neurons per hidden layer and ReLU activations. This architecture, typical in WGAN implementations \cite{gulrajani_improved_2017, arjovsky_wasserstein_2017}, provides sufficient capacity for our tasks. Based on previous theoretical results, we formulate two hypotheses below:
\begin{hyp}
For a fixed quantum resource budget, the \textit{Asynchronous} and \textit{Decoupled} approaches typically achieve better performance than the \textit{Joint} approach.
\label{hyp: faster convergence}
\end{hyp}

\begin{hyp}
When training OT-EVS in the shadow-frugal parameterization, the classical shadows approach achieves comparable performance with fewer measurements than the conventional measurement scheme.
\label{hyp: shadow better}
\end{hyp}

To verify these hypotheses, we first conduct controlled experiments where the training dataset is generated by an OT-EVS with random parameters, after which an identically structured OT-EVS (randomly initialized) learns to reproduce the distribution, following the methodology of Ref. \cite{romero_variational_2021}. This setting ensures that the models are sufficiently expressive to represent the target distribution exactly, providing an ideal testbed for comparing training methods and analyzing shot noise effects. We systematically evaluate all combinations of the three training methods, two measurement methods, and nine distinct measurement budgets, with each configuration tested over $20$ independent trials. Model performance is quantified by the Kullback-Leibler (KL) divergence after $50,000$ training iterations. 
\begin{figure}[t]
  \begin{center}    \centerline{\includegraphics[width=0.8\columnwidth]{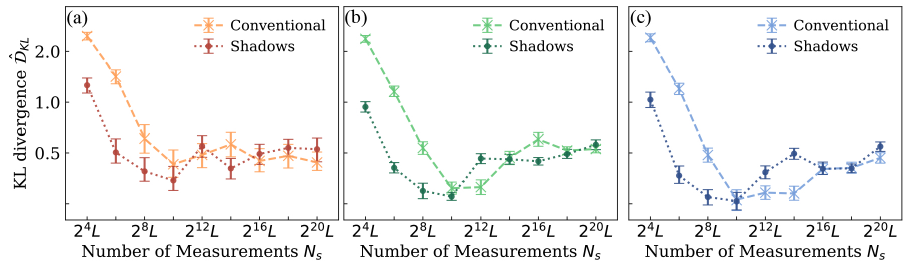}}
    \caption{Training performance of OT-EVS with a two-layer $8$-qubit \textit{sequential} circuit using the (a) \textit{Joint} (b) \textit{Asynchronous} (c) \textit{Decoupled} method on the synthetic dataset. The interquartile mean and bootstraped $95\%$ confidence intervals over $20$ trials for the estimated KL divergence after $50k$ training iterations are shown.}
    \label{figure1}
  \end{center}
\end{figure}
Figure \ref{figure1} presents results for an OT-EVS with a two-layer $8$-qubit sequential circuit and $2$-local observables, generating $8$-dimensional data (results for additional architectures are in the Supplementary Notes). We observe that insufficient measurement budgets impair model convergence and that classical shadows achieve performance comparable to the conventional scheme with at least four times fewer measurements. In addition, for most measurement settings, the \textit{Asynchronous} and \text{Decoupled} methods outperform the \textit{Joint} method, supporting Hypothesis \ref{hyp: faster convergence} and Hypothesis \ref{hyp: shadow better}. Interestingly, optimal performance occurs at intermediate measurement budgets. While this phenomenon appears consistently in our experiments, its generalizability requires further investigation. We hypothesize that this non-monotonic pattern arises from an implicit regularization effect: 
\begin{hyp}
Moderate shot noise (measurement error) improves the alignment of the generated and the target distributions in dimensionality, thereby stabilizing the training, analogous to the beneficial role of noise injection in classical GANs \cite{salimans_improved_2016, brock_large_2018, karras_style-based_2021, sonderby_amortised_2016, arjovsky_towards_2022, roth_stabilizing_2017, feng_understanding_2021}.
\end{hyp}
Next, we systematically compare OT-EVS and OF-EVS across different architectural configurations to assess the role of tunable observables. The training datasets comprise $8$-dimensional data generated by OT-EVS models with one-layer $8$-qubit (sequential or brickwork) circuits and $1$-local general observables. For each dataset, we train models with the same circuit ansatz but vary eight circuit depths and three observable settings: tunable observables matching those used to generate the dataset but randomly initialized, fixed observables matching those used to create the dataset, and the fixed $1$-local Pauli-Z observables (a standard choice in EVS literature \cite{romero_variational_2021, anand_noise_2021}), supplemented with global scaling and translation parameters to postprocess raw outputs to compensate the natural boundedness of Pauli-Z. Each configuration is evaluated over $20$ independent trials, with results shown in Figure \ref{figure2}.
\begin{figure}[h]
  \begin{center}    \centerline{\includegraphics[width=0.8\columnwidth]{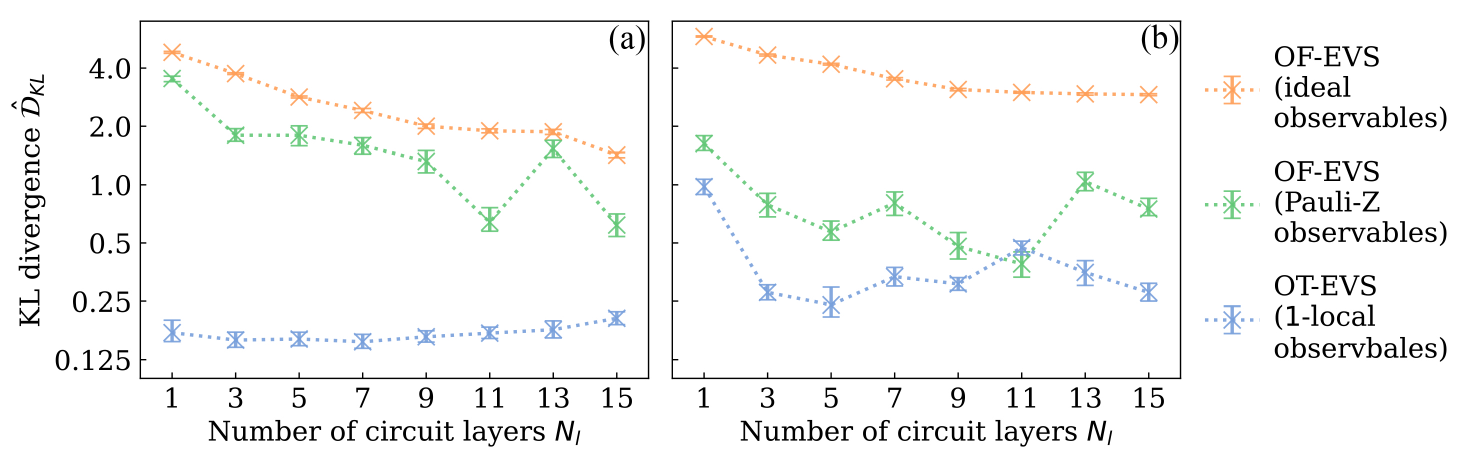}}
    \caption{Training performance of OT-EVS and OF-EVS with varying numbers of circuit ansatz layers on synthetic datasets. (a) Sequential ansatz. (b) Brickwork ansatz. The interquartile mean and bootstraped $95\%$ confidence intervals over $20$ trials for the estimated KL divergence after $50k$ training iterations are shown.}
    \label{figure2}
  \end{center}
\end{figure}
Contrary to concerns about the \textit{barren plateau} phenomenon \cite{mcclean_barren_2018, ragone_lie_2024} in deep PQCs, we observe no apparent degradation in training performance with increasing circuit depth. We note that both ansatzes are designed such that neighboring layers can cancel each other as identity for specific parameter values, ensuring a monotonic increase of expressivity with depth. OT-EVS exhibits a clear advantage: models with tunable observables consistently achieve lower KL divergence values than their fixed-observable counterparts at equivalent depths. Surprisingly, the added learning of weight matrices in OT-EVS facilitates training rather than hinders it. The persistent performance gap between OF-EVS models with different observable choices underscores that circuit expressivity alone cannot fully compensate for the flexibility of tunable observables. 

We further evaluate the models on the MNIST dataset \cite{lecun_gradient-based_1998}, comprising $60,000$ grayscale images ($32\times32$ pixels). Models are trained to learn $8$-dimensional representations of the data from a pretrained autoencoder \cite{rumelhart_parallel_1986, bank_autoencoders_2021}, a standard machine learning technique. We use a two-layer $8$-qubit brickwork circuit with either tunable $1$-local general observables (trained via the \textit{Decoupled} method with classical shadows or conventional measurements), or fixed $1$-local Pauli-Z observables (with global scaling and translation), under varying measurement budgets. Model performance after $5000$ training iterations (averaged over $5$ independent trials) and generated samples are shown in Figure \ref{figure3}. Consistent with controlled experiments, best performances for both models are achieved at intermediate measurement budgets, and classical shadows reduce the optimal measurement budget of the OT-EVS by at least fourfold. 

\begin{figure}[h]
  \begin{center}    \centerline{\includegraphics[width=0.8\columnwidth]{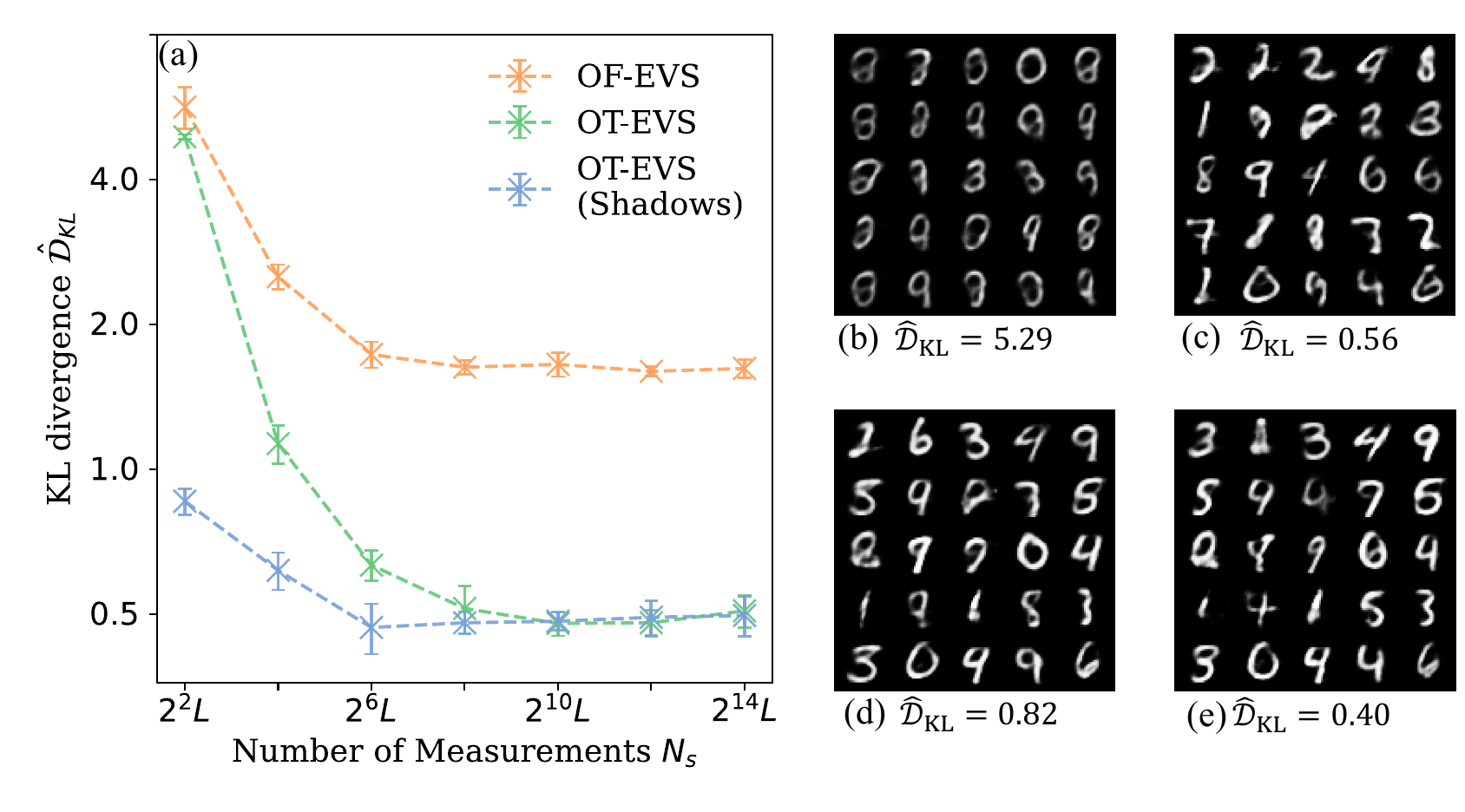}}
    \caption{(a) Training performance of OT-EVS with $1$-local general observables (using the \textit{Decoupled} method) and OF-EVS with $1$-local Pauli-Z observables with different numbers of measurements on the MNIST dataset. The mean and standard deviation over $5$ trials are shown. Samples generated by an $1$-local general observable OT-EVS trained with $2^2L$ (b) and $2^{10}L$ (c) conventional measurements. Samples generated by an $1$-local general observable OT-EVS trained with $2^2L$ (d) and $2^{10}L$ (e) classical shadows.}
    \label{figure3}
  \end{center}
\end{figure}

Finally, we examine whether enlarging the set of observables improves OT-EVS performance. We benchmark on both MNIST and the more challenging Fashion-MNIST dataset \cite{xiao_fashion-mnist_2017}, embedded in $32$-dimensional representations by a pretrained autoencoder. We use two-layer brickwork circuits with three different tunable observable choices: $1$-local Pauli-Z observables, $1$-local general observables or $2$-local general observables (trained via the \textit{Decoupled} method with $2^{10}L$ conventional measurements). For a classical baseline we train a three-layer MLP with $512$ neurons per hidden layer and ReLU activations, which mirrors the critic network architecture. To test scalability, we investigate the performance under different input dimensions $K$ for the latent variables $\vec{z}$ (Definition \ref{def:evs}), which equals the qubit number in the brickwork circuit, and report performance after $5000$ training iterations, averaged over five independent trials (Figure \ref{figure4}). On both datasets and across all input dimensions, the $2$-local OT-EVS consistently outperforms the $1$-local OT-EVS, which in turn outperforms the Pauli-Z OT-EVS, demonstrating that richer observable sets enhance model expressivity. The $2$-local OT-EVSs also achieve performances comparable with the classical MLP on both datasets. Nevertheless, we caution against interpreting these results as evidence of quantum advantage for larger problem instances, which would require further investigation into EVS circuit architectures and more comprehensive comparisons with a broader range of classical neural network models beyond MLPs.

\begin{figure}[h!]
  \begin{center}    \centerline{\includegraphics[width=0.8\columnwidth]{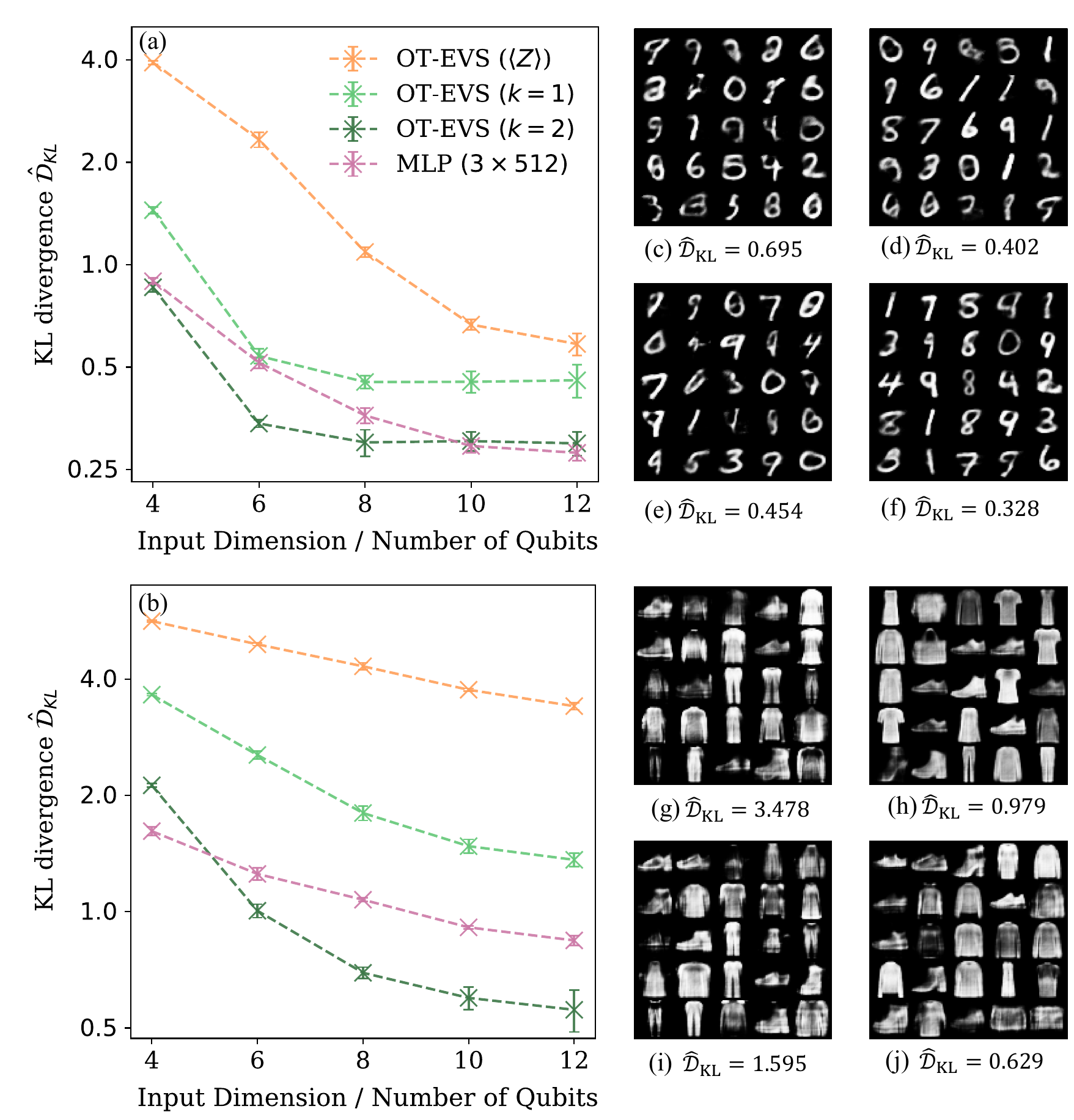}}
    \caption{Training performance of OT-EVS with $1$-local Pauli-Z observables, $1$-local general observables or $2$-local general observables (using the \textit{Decoupled} method) and a three-layer $512$-neuron MLP on the MNIST dataset (a) and Fashion-MNIST dataset (b). Samples generated by a $12$-qubit $1$-local Pauli-Z observable OT-EVS (c), a $12$-input MLP (d), a $12$-qubit $1$-local general observable OT-EVS (e), and a $12$-qubit $2$-local general observable OT-EVS (f) for the MNIST dataset. Samples generated by a $12$-qubit $1$-local Pauli-Z observable OT-EVS (g), a $12$-input MLP (h), a $12$-qubit $1$-local general observable OT-EVS (i), and a $12$-qubit $2$-local general observable OT-EVS (j) for the Fashion-MNIST dataset.}
    \label{figure4}
  \end{center}
\end{figure}

\section{Discussion}
\label{sec:discussion}
This work proposes a quantum generative model extending the Expectation Value Sampler using tunable observables. 
The model outputs random variables, given by the expectation values of an internal random state. 
Ideally, the outcomes follow a target probability distribution accessible through data. 
The tunable-observable extension enhances the expressivity of the model while reducing the quantum resources required for training, compared to fixed-observable models. These improvements are demonstrated both analytically and empirically. Our contributions are as follows. 
a) The expressivity enhancement is demonstrated by showing that fixed-observable models are included in tunable-observable models, not the converse. 
b) We can improve sample complexity by restricting our observables to those efficiently estimated using classical shadows. 
We also provide upper bounds on the number of measurements required to estimate the output probability distribution up to a predefined precision.
c) We propose two new adversarial training procedures with improved performances while maintaining the same requirements of quantum resources, by fine-tuning the update rate of parameters in the observables and the quantum circuits.

We showcase the performance of tunable-observable models by conducting numerical experiments. First, tunable-observable models outperform fixed-observable models in practice, even in data generated by fixed-observable models. Second, the tailored training procedures and the classical-shadow sampling strategy reduce the overall training cost. 
Furthermore, the training procedures exhibit an interesting phenomenon: moderate shot noise levels can improve training performance, similar to the noise injection in classical GAN training. 

Interesting foreseeable research directions include investigating the performance of tunable-observable models on low-dimensional, real-world data sets. Technically valuable contributions are a theoretical analysis of the convergence properties and stability of the proposed training procedures and a numerical characterization of the shot noise level that benefits the model's performance.

\section{Methods}\label{sec:methods}
\subsection*{Measurement Protocols for Expectation Value Estimation}
\label{subsec: eve}
The expectation values $(\braket{\psi|O_l|\psi})^{L}_{l=1}$ can be estimated through different strategies. In the most straightforward approach, the conventional method, each observable $O_l$ is measured independently. This involves preparing $N_c$ copies of $\ket{\psi}$, applying a circuit that diagonalizes $O_l$ in the device's measurement basis to each copy, and performing a projective measurement. The expectation value is unbiasedly estimated as the sample mean of the outcomes, with statistical accuracy scaling as $\mathcal O_p(1/\sqrt{N_c})$. 

For the shadow-frugal parameterization, where all $(\braket{O_l})^{L}_{l=1}$ are $k$-local Pauli strings, more sophisticated techniques may be employed. These include importance sampling, which prioritizes measurements of higher-weight Pauli strings \cite{wecker_towards_2015, mcclean_theory_2016}, and simultaneous measurements protocols that exploit groupings of commuting observables \cite{mcclean_theory_2016, kandala_hardware-efficient_2017}. While optimal grouping for maximal sample efficiency is known to be NP-hard, heuristic approaches have been extensively explored in literature \cite{gokhale_minimizing_2019, zhao_measurement_2020, verteletskyi_measurement_2020, hamamura_efficient_2020, crawford_efficient_2021}. However, the dynamic nature of the weight matrix $\vec{\alpha}$ during OT-EVS training introduces additional complexity, as it would necessitate adaptive grouping strategies. For this reason, we exclude such advanced methods from both Theorem \ref{thm:sample_complexity} and our numerical experiments.

An alternative strategy for the shadow-frugal parameterization is the classical shadows \cite{huang_predicting_2020}. Given a quantum state $\ket{\psi}$, the procedure for estimating $(\braket{O_l})^{L}_{l=1}$ via classical shadows proceeds as follows: First, a unitary $U=U_1\otimes U_2\cdots\otimes U_n$ is randomly selected from the Pauli ensemble, where each $U_j$ acts on a single qubit, and applied to $\ket{\psi}$. A computational basis measurement is then performed, yielding an outcome $\hat{b}\in \{0,1\}^n$. Letting $\rho = \ket{\psi}\bra{\psi}$, the average effect of this process is described by a quantum channel:
\begin{equation}
    \mathcal{M}(\rho) = \mathbb E\left[ U^\dagger |\hat{b}\rangle\langle \hat{b}| U \right].
\end{equation}
For a specific outcome $\hat{b}$, the classical snapshot $\hat{\rho}(\hat b)$ is obtained by inverting this channel:
\begin{align}\label{eq.shadow}
\hat{\rho}(\hat b) \coloneq \mathcal M^{-1}\left( U^\dagger |\hat{b}\rangle\langle \hat{b}| U \right)=\bigotimes_{j=1}^n \left(3U_j^\dagger |\hat{b}\rangle\langle \hat{b}| U_j-\mathbb{I}\right). 
\end{align}
This estimator is unbiased, satisfying $\mathbb{E}[\hat{\rho}(\hat{b})]=\rho$, and can be efficiently postprocessed to estimate all expectation values simultaneously. 
Further improvements in sample efficiency can be achieved by incorporating knowledge of the weight matrix $\vec{\alpha}$, such as using a locally-biased measurement distribution \cite{hadfield_measurements_2022} or derandomization techniques \cite{huang_efficient_2021}. However, these techniques are excluded from Theorem \ref{thm:sample_complexity} and our numerical experiments for simplicity.

\subsection*{Numerical Simulation of Measurement Protocols}
\label{app: shot noise}
We develop an original method to efficiently simulate conventional measurements and classical shadows in our numerical experiments. Let $\vec{y} = \left(\braket{A_1}, \dots, \braket{A_M}\right)^T$ represent an ideal sample that would be generated by OT-EVS in the absence of measurement errors. Since both the conventional measurements and the classical shadows provide unbiased estimators, we model the shot noise in either case as a centered Gaussian random vector $\vec{\epsilon} \sim \mathcal{N}(0, \Sigma/N_s)$ added to the measured observables $\left(O_l\right)^{L}_{l=1}$. Since measurement outcomes are independent and identically distributed, the normal approximation becomes accurate for sufficiently large $N_s$ due to the Central Limit Theorem. The shot noise $\vec{\epsilon}$ depends on the measured state $\ket{\psi}$ (and consequently on $\vec{z}$ and $\vec{\theta}$), but not on $\vec{\alpha}$. The noise-perturbed sample can then be expressed as:
\begin{equation}
\tilde{\vec{y}} = \vec{y} + \vec{\alpha} \vec{\epsilon}
\end{equation}

When computing the derivative of the generator loss (equation (\ref{eq:lg})) with respect to $\vec{\vec{\theta}}$, as an example, we employ the reparameterization trick \cite{kingma_auto-encoding_2022} to enable automatic differentiation. This involves decomposing the stochastic node $\vec{\epsilon}$ as:
\begin{equation}
\vec{\epsilon} = \frac{1}{N_s}S \vec{\xi}, \end{equation}
where $S$ is the lower triangular matrix from the Cholesky decomposition of the covariance matrix $\Sigma$ (i.e., $\Sigma=SS^T$) and $\vec{\xi} = \left( (\xi_l)_{l=1}^L \right)^T$ is a vector of standard Gaussian random variables. Applying the chain rule yields:
\begin{equation}
    \frac{\partial \mathcal{L}_G}{\partial \vec{\vec{\theta}}} = \frac{1}{B}\sum_{b=1}^B \left[\frac{\partial D_{\vec{w}}}{\partial \tilde{\vec{y}}^b} \left( \frac{\partial \vec{y}^b}{\partial \vec{\vec{\theta}}} + \frac{1}{N_s} \vec{\alpha} \frac{\partial S^b}{\partial \vec{y}^b} \frac{\partial \vec{y}^b}{\partial \vec{\vec{\theta}}} \vec{\xi}^b \right) \right].
\end{equation}
This formulation isolates the stochasticity in the non-parameterized nodes $\vec{\xi}^b$, allowing all computational steps to be handled by automatic differentiation. The remaining task is to derive the covariance matrix $\Sigma$ for each measurement scheme. 

For conventional measurements, where each $O_l$ is measured independently, $\Sigma$ takes a simple diagonal form: \begin{equation}
\Sigma^2_{ll}= 1- \braket{O_l}^2.
\end{equation} 
This diagonal covariance resembles the approximation method in Ref. \cite{bonet-monroig_performance_2023}, though we note that while their approach uses the exact binomial distribution, we employ the Gaussian approximation for compatibility with automatic differentiation. 

In contrast, the classical shadows produce a non-diagonal $\Sigma$ due to correlations in estimating different Pauli strings. Specifically, a measurement performed in the basis $Q=\bigotimes_{\iota=1}^n Q^\iota$ contributes to the estimation of $\braket{O_l}=\braket{\bigotimes_{\iota=1}^n O_l^\iota}$ when $O_l^\iota \in \{O_l^\iota, \mathds{1} \}, \forall 1 \leq \iota \leq n$, and not otherwise. We derive the elements of $\Sigma$ using results from Ref. \cite{huang_predicting_2020} (see equations (S9), (10), (35), (36), (50) of the reference): 
\begin{equation}
\begin{aligned}
\Sigma_{i, j}^2
&= \underset{U \sim \mathrm{Cl}(2^n)}{\mathbb{E}} \sum_{x\in \{0,1\}^n} \abs{\bra{x} U \ket{\psi_{\vec{\vec{\theta}}}(\vec{z})}}^2 \bra{x}U \mathcal{M}^{-1}(O_i)U^\dagger \ket{x}\bra{x}U \mathcal{M}^{-1}(O_j)U^\dagger \ket{x} - o_i o_j  \\
&= \left( \prod_{\iota=1}^n f(O^\iota_i, O^\iota_j) \right) o_{ij} - o_i o_j,
\end{aligned}
\end{equation}
where $o_i = \braket{O_i}$ and $o_{ij} = \braket{O_i O_j}$, and the function $f$ is defined as:
\begin{equation}
f(O^\iota_i, O^\iota_j) = \begin{cases} 
      0 & \text{if } (O_i^\iota \neq O_j^\iota) \text{ and } (O_i^\iota \neq \mathbb{I}) \text{ and }  (O_i^\iota \neq \mathbb{I})\\
      1 & \text{if } ((O_i^\iota \neq O_j^\iota) \text{ and } ((O_i^\iota = \mathbb{I}) \text{ or } (O_j^\iota = \mathbb{I}))) \text{ or } (O_i^\iota = O_j^\iota = \mathbb{I})\\
      3 & \text{if } (O_i^\iota = O_j^\iota \neq \mathbb{I})
       \end{cases}
\end{equation}
We note that the observables $O_i O_j$ required for this calculation are typically not included in the original set $(O_l)^{L}_{l=1}$, introducing additional simulation overhead for classical shadows compared to conventional measurements.

\begin{figure}[t]
  \begin{center}    \centerline{\includegraphics[width=0.8\columnwidth]{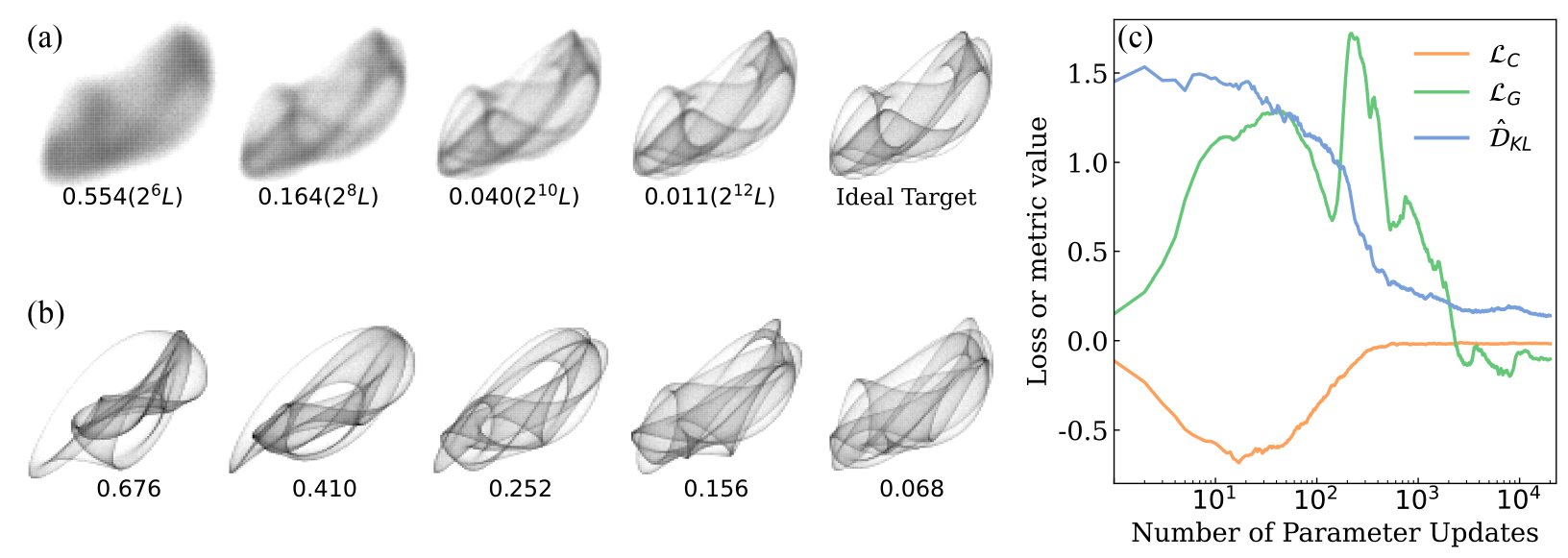}}
    \caption{(a) Noise-perturbed and ideal empirical densities of an example $2D$ synthetic dataset with estimated KL divergence (and measurement budget) labelled. (b) Generated distributions at selected training steps (ignoring shot noise), with estimated KL divergence labelled. (c) Example training curves (generator loss, critic loss, and estimated KL divergence) for a training using the \textit{Asynchronous} method.}
    \label{figure5}
  \end{center}
\end{figure}

\subsection*{Model Performance Evaluation}
In our numerical experiments, we employ the KL divergence as the evaluation metric for EVS models in our numerical experiments, a popular choice in the assessment of generative models \cite{borji_pros_2022, dinh_density_2016, ho_flow_2019, song_score-based_2020}: 
\begin{definition}[Kullback-Leibler divergence~\cite{kullback_information_1951}]
    Assume a probability space $(\mathcal{X}, \mathcal{A}, \mu)$. The Kullback-Leibler divergence between two absolutely continuous (w.r.t. $\mu$) probability distributions $\mathbb{P}$ and $\mathbb{Q}$ is defined as:
\begin{equation}
    \mathcal{D}_{KL}(\mathbb{P} \parallel \mathbb{Q}) = \int_{\mathcal{X}} p(x)\log\left(\frac{p(x)}{q(x)}\right) \, d\mu(x),
\end{equation}
where $p=\frac{\mathrm{d}\mathbb{P}}{\mathrm{d}\mu}$ and $q=\frac{\mathrm{d}\mathbb{Q}}{\mathrm{d}\mu}$ are the densities of $\mathbb{P}$ and $\mathbb{Q}$. 
\end{definition}

In practice, the exact densities of both the generated and target distributions are unknown. To estimate the KL divergence, we adopt an estimator based on K-nearest-neighbors \cite{wang_divergence_2009} (see equation (25) of the reference). This approach requires only samples from the training and generated distributions and provides high accuracy even for moderately high-dimensional distributions (e.g., with dimensions up to $\sim 100$).

\subsection*{Methods for Controlled Experiments}
In our controlled experiments, we generate each training dataset using a randomly initialized instance of an OT-EVS model. When preparing a dataset, the quantum circuit parameters $\vec{\theta}$ are individually drawn from a Gaussian distribution $\mathcal{N}(0, \pi/8)$, followed by a global random rotation with an angle uniformly drawn from $ U[-\pi, \pi)$. The weight matrix $\vec{\alpha}$ is structured so that each row contains $1,4,9$ at random positions, with all other elements set to $0$. This configuration ensures sufficient variability between trials while maintaining controlled experimental conditions. All datasets contain exactly $4096$ samples. 

Our synthetic datasets exhibit non-trivial statistical structures due to their projection from higher-dimensional observable spaces (with the data dimension intentionally kept smaller than the number of observables). To provide concrete insight into these distributions, we analyze a $4$-qubit OT-EVS, computing its $12$ one-local expectation values and projecting them onto a $2$-dimensional space.
Figure \ref{figure5}(a) shows the empirical densities of these projected values as histograms, comparing the ideal (infinite-measurement) with noise-perturbed (finite-measurement) distributions.
As the measurement budget increases (left to right), the generated distributions develop sharper structures, aligning well with the decreasing KL divergence between the ideal with noise-perturbed distributions. 

We train the OT-EVS on the prepared ideal dataset using the \textit{Asynchronous} method and track the evolution of the learned distributions (with infinite measurements) in Figure \ref{figure5}(b). Figure \ref{figure5}(b) plots the corresponding training curves for the generator loss, critic loss, and estimated KLD. The critic loss rapidly converges to zero, even as the generated distributions keep improving. This behavior aligns with the observations in Ref. \cite{arjovsky_wasserstein_2017} (Section 4.2), confirming that the critic loss, while useful for training, is not a good evaluation metric. In contrast, the nearly monotonic decrease of the estimated KLD closely matches the improving visual similarity between the generated and target distributions, validating its use as an effective evaluation metric in our experiments.

\section{Declarations}
\paragraph{Data Availability}
All data generated or analysed during this study can be easily reproduced using our documented code.

\paragraph{Code Availability}
\par The underlying code for this study can be accessed via this link \url{https://github.com/kevinkayyy/OT-EVS_Quantum_Generative_Model}.

\paragraph{Competing interests} The authors declare no competing interests. 

\paragraph{Author Contributions}
All authors contributed to the ideation and initial discussions. K.S., A.K., and V.D., developed the theoretical framework for model definition, observable parameterization, usage of classical shadows and expressivity analysis. K.S., A.K., and H.W. performed resource analysis and designed the training algorithm. K.S. designed the numerical simulation technique for classical shadows, developed the code, performed numerical experiments and analyzed the results.  V.D. and H.W. supervised the work.  K.S. and A.P.-S. wrote the manuscript and all authors contributed to the discussion of results and the review of the manuscript.

\paragraph{Acknowledgements}
KS thanks Riccardo Molteni, Yash Patel, Susanne Pielawa, and Peter Wegmann for insightful discussions. This research was funded by the BMW Group. 
VD thanks the support of the Dutch National Growth Fund (NGF), as part of the Quantum Delta NL programme. 
This work was also partially supported by the Dutch Research Council (NWO/OCW), as part of the Quantum Software Consortium programme (project number 024.003.03), and co-funded by the European Union (ERC CoG, BeMAIQuantum, 101124342).

\newpage
\bibliography{ref}

\newpage 
\section*{Supplementary Note 1 -- Model Expressivity}
We extend the main text's discussion on the expressivity of the Observable-Tunable Expectation Value Sampler (OT-EVS). We verify the validity of the two examples given in the main text.

In Example 1, outputs of the OF-EVS take the form: 
\begin{equation}
y_m = \sin\left( \frac{\theta_m + z_m}{2}\right) \hspace{0.3cm} \text{for } m\in \{1,2 \},
\end{equation} 
while outputs of the OT-EVS take the form $(\vec{\alpha}_{11}y_1 + \vec{\alpha}_{12}y_2, \vec{\alpha}_{21}y_1 + \vec{\alpha}_{22}y_2)$. We can attribute the increase in expressivity to two causes. First, linear combinations of $y_1$ and $y_2$ follow different distributions than $y_1$ or $y_2$ in general, even though $y_1$ and $y_2$ are identically distributed up to $\vec{\theta}$. Second, $\vec{\alpha}_{11}y_1 + \vec{\alpha}_{12}y_2$ is in general correlated with $\vec{\alpha}_{21}y_1 + \vec{\alpha}_{22}y_2$, while $y_1$ and $y_2$ themselves are always independent. 

To verify Example 2, we utilize existing theoretical tools \cite{bonet-monroig_verifying_2024}. When measuring a Haar-random state, the expectation value of a given observable $O$ is a random variable derived from the inner product of a symmetric Dirichlet random variable and the eigenvalues of the observable. Consider an observable $O_{\vec{\alpha}}$ defined by weighted sums of $O$ and some other observables. The eigenvalues of $O_{\vec{\alpha}}$ will depend on the spectra of all constituent observables and their weights. Consequently, the resulting spectrum and the underlying distribution of the expectation value differ from those of $O$ originally. 

\section*{Supplementary Note 2 -- Proof of Theorem~\ref{thm:sample_complexity} of the main text}
First, we denote by $\vec{y}^i$ and $\tilde{\vec{y}}^i$ the $i$-th sample (in a batch) generated by the shot-noise-free and shot-noise perturbed OT-EVS in shadow-frugal parameterization, respectively: 
\begin{align}
            \vec{y}^i =& \left[\tr(H_1 \rho^i), \dots, \tr(H_M \rho^i)\right]^\top = \left[\sum_{l=1}^L \vec{\alpha}_{1l} \langle O_l^i \rangle, \cdots,  \sum_{l=1}^L \vec{\alpha}_{Ml} \langle O_l^i \rangle\right]^\top = \vec\alpha \left[\langle O_1^i \rangle, \cdots, \langle O_L^i \rangle\right]^\top = \vec\alpha \langle \vec{O}^i \rangle. \\
            \tilde{\vec{y}}^i =& \left[\hat{h}^i_1, \dots, \hat{h}^i_M\right]^\top = \left[\sum_{l=1}^L \vec{\alpha}_{1l} \widehat{\langle O_l^i \rangle}, \cdots,  \sum_{l=1}^L \vec{\alpha}_{Ml} \widehat{\langle O_l^i \rangle}\right]^\top= \vec{\alpha} \left[\widehat{\langle O_1^i \rangle}, \cdots, \widehat{\langle O_L^i \rangle}\right]^\top = \vec\alpha \widehat{\langle \vec{O}^i \rangle},
\end{align}
where $\widehat{\langle O_l^i \rangle}$ is the estimation of the expectation value of the $l$-th Pauli string for the $i$-th sample by a particular measurement strategy. Next, we prove two lemmas required for Theorem~\ref{thm:sample_complexity}.
\begin{lemma} \label{probability_uneq}
    Let $L$ denote the number of $k$-local Pauli strings, and let $B$ represent the batch size. Denote $\mathbb{P}^B$ and $\tilde{\mathbb{P}}^B$ as the empirical distributions of the shot-noise-free and shot-noise-perturbed samples, respectively, generated by an OT-EVS under the shadow-frugal parameterization specified by a weight matrix $\vec\alpha$ and let $\norm{\vec\alpha}_{\infty} \leq T$. Then, for any $\epsilon > 0$, the following inequality can be established for the $\mathcal{W}_1$-distance between the shot-noise-free distribution $\mathbb{P}^B$ and the shot-noise-perturbed distribution $\tilde{\mathbb{P}}^B$:
    \begin{equation}
        \mathrm{Pr} (\mathcal{W}_1(\mathbb{P}^B, \tilde{\mathbb{P}}^B) \geq \epsilon)
            \leq \mathrm{Pr}\left(\bigcup_{i=1}^N \bigcup_{l=1}^L \left\{\lvert \langle O_l^i \rangle - \widehat{\langle O_l^i \rangle}\rvert \geq \frac{\epsilon}{T}\right\}\right).
    \end{equation}
\end{lemma}

\begin{proof}
Using definition of $\mathcal{W}_1$-distance and properties of the norm we have:
    \begin{align}
            \mathcal{W}_1(\mathbb{}^B, \tilde{\mathbb{P}}^B) &= \inf_{\pi\in S_N} \left( \frac{1}{N} \sum_{i=1}^N \lVert \vec{y}^i - \tilde{\vec{y}}^{\pi(i)}\rVert_1 \right) \leq \frac{1}{N} \sum_{i=1}^N \lVert \vec{y}^i - \tilde{\vec{y}}^i\rVert_1\\
            &= \frac{1}{N} \sum_{i=1}^N \lVert \vec\alpha (\langle  \vec{O}^i  \rangle - \widehat{\langle \vec{O}^i \rangle})\rVert_1 \leq \frac{1}{N} \sum_{i=1}^N \lVert \vec\alpha \rVert_{\infty} \lVert \langle  \vec{O}^i  \rangle - \widehat{\langle \vec{O}^i \rangle} \rVert_{\infty} \\
            &\leq \frac{T}{N} \sum_{i=1}^N \max_l \lvert \langle  O^i_l  \rangle - \widehat{\langle O^i_l \rangle} \rvert.
    \end{align}
Taking into account the probability inequality:
    \begin{equation}
        \begin{aligned}
            \mathrm{Pr} (\mathcal{W}_1(\mathbb{P}^B, \tilde{\mathbb{P}}^B) \geq \epsilon) \leq \mathrm{Pr} \left(\frac{1}{N}\sum_{i=1}^N\max_{l}\lvert \langle O_l^i \rangle - \widehat{\langle O_l^i \rangle}\rvert \geq \frac{\epsilon}{T}\right) \leq \mathrm{Pr} \left(\bigcup_{i=1}^N \bigcup_{l=1}^L \left\{\lvert \langle O_l^i \rangle - \widehat{\langle O_l^i \rangle}\rvert \geq \frac{\epsilon}{T}\right\}\right).
        \end{aligned}
    \end{equation}
\end{proof}
\begin{lemma}[restatement of Th1, p.13 and L3, p.26 from \cite{huang_predicting_2020}]\label{huang_restatement}
    Fix a measurement primitive $\mathcal{U}$ of randomly chosen Pauli measurements, a collection of observables $P_1,\ldots,P_L$, which are at most k-local Pauli strings. Fix accuracy parameters $\epsilon,\delta \in [0,1]$. Then, a collection of
    \begin{equation}
        K = \left\lceil 68\frac{3^k}{\epsilon^2} \log \left(\frac{2L}{\delta}\right) \right\rceil
    \end{equation}
    independent classical shadows allow for accurately predicting expectation  values of all observables via the median of means prediction, which means that:
    \begin{equation}
        \left| \hat{p}_l (K) - \mathrm{tr} \left( P_l \rho \right) \right| \leq \epsilon \quad \text{for all $1 \leq l \leq L$}
    \end{equation}
    with probability at least $1-\delta$.
\end{lemma}
Next, we show the main theorem again and prove it with the above lemmas.
\begin{theorem}[Sample Complexity]\label{thm:sample_complexity}  
Let $L$ be the number of $k$-local Pauli strings. Let $B$ be the batch size. Denote by $\mathbb{P}^B$ the empirical noiseless distribution of a batch of samples generated from an OT-EVS with shadow-frugal parameterization, and by $\tilde{\mathbb{P}}^B$ the distribution with measurement error. We assume $\norm{\vec\alpha}_{\infty} \leq T$, implying $\norm{A_m}_{\infty}\leq T$, $\forall 1\leq m \leq M$, where $A_m = \sum_{l = 1}^L \vec\alpha_{m, l} O_l$. For $\epsilon, \delta >0$, the following probabilistic bound
\begin{equation}
\operatorname{Pr} \left(\mathcal{W}_1(\mathbb{P}^B, \tilde{\mathbb{P}}^B) \leq \epsilon \right) > 1- \delta
\end{equation}
holds with measurements $N_s$,
\begin{align}
    N_s &\geq 68 \cdot3^k \left\lceil B \frac{T^2}{\epsilon^2} \log \left( \frac{2LB}{\delta} \right) \right\rceil \text{ (with classical shadows)} \label{shadow equation} \\
    N_s &\geq 2 L \left\lceil B \frac{T^2}{\epsilon^2} \log \left( \frac{2LB}{\delta} \right) \right\rceil  \label{conventional equation} \text{(with the conventional measurements)}
\end{align}
\end{theorem}
\begin{proof}
\emph{Conventional Measurements:}
Using Lemma~\ref{probability_uneq} and Boole's inequality, we can place an upper bound on the probability for the expression involving the $\mathcal{W}_1$-distance between distributions and then impose a condition that this probability is sufficiently small:
    \begin{equation}
        \begin{aligned}
                \mathrm{Pr} (\mathcal{W}_1(\mathbb{P}^B, \tilde{\mathbb{P}}^B) \geq \epsilon) \leq     \mathrm{Pr} \left(\bigcup_{i=1}^B \bigcup_{l=1}^L \left\{\lvert \langle O_l^i \rangle - \widehat{\langle O_l^i \rangle}\rvert \geq \frac{\epsilon}{T}\right\}\right) \leq \sum_{i=1}^B  \sum_{l=1}^L  \mathrm{Pr} \left( \left| \hat{o}_l^i (K) - \mathrm{tr} \left( O_l \rho^i \right) \right| \geq \frac{\epsilon}{T} \right) \leq \delta.
        \end{aligned}
    \end{equation}
    For any fixed $i \in [B]$ and $l \in [L]$ we compute average over $K$ measurements of $\rho_i$, where $K$ is chosen such that:
    \begin{equation}
        \left| \hat{o}_l^i (K) - \mathrm{tr} \left( O_l \rho^i \right) \right| \leq \frac{\epsilon}{T} \quad \text{for all $1 \leq i \leq B$, $1 \leq l \leq L$}
    \end{equation}
    with failure probability $\frac{\delta}{BL}$. Using Hoeffding inequality, we can compute $K$:
    \begin{equation}
            \mathrm{Pr} \left(\left| \hat{o}_l^i (K) - \mathrm{tr} \left( O_l \rho^i \right) \right| \geq \frac{\epsilon}{T} \right) \leq 2 \exp \left\{-\frac{2K^2 \epsilon^2}{4KT^2} \right\} \leq \frac{\delta}{BL}.
    \end{equation}
    Thus, $K = \left\lceil 2 \frac{T^2}{\epsilon^2} \log \left( \frac{2BL}{\delta} \right) \right\rceil$. Multiplying it by $BL$ gives us the total number of measurements.
\\ \\
\hspace{-0.33cm}\textit{Classical Shadows:}
Using Lemma~\ref{probability_uneq} and Boole's inequality, we can place an upper bound on the probability for the expression involving the $\mathcal{W}_1$-distance between distributions and then impose a condition that this probability is sufficiently small:
    \begin{align}
            \mathrm{Pr} (\mathcal{W}_1(\mathbb{P}^B, \tilde{\mathbb{P}}^B) \geq \epsilon) \leq \mathrm{Pr}\left(\bigcup_{i=1}^B \bigcup_{l=1}^L \left\{\lvert \langle O_l^i \rangle - \widehat{\langle O_l^i \rangle}\rvert \geq \frac{\epsilon}{T}\right\}\right)
             \leq \sum_{i=1}^B \mathrm{Pr}\left( \bigcup_{l=1}^L \left\{ \left| \hat{o}_l^i (K) - \mathrm{tr} \left(O_l^i \rho^i \right) \right| \geq \frac{\epsilon}{T} \right\} \right) \leq \delta
    \end{align}
    For any fixed $i \in [B]$, the number of measurements $K$ for constructing shadow array is chosen such that:
    \begin{equation}
        \left| \hat{o}_l (K) - \mathrm{tr} \left( O_l \rho \right) \right| \leq \frac{\epsilon}{T} \quad \text{for all $1 \leq l \leq L$}
    \end{equation}
    with failure probability $\frac{\delta}{B}$.
    
    Using Lemma~\ref{huang_restatement} we can conclude that $K =  \left\lceil 68 \frac{T^2 3^k}{\epsilon^2} \log \left( \frac{2BL}{\delta} \right) \right\rceil$. Multiplying it by $B$ gives us a total number of measurements.
\end{proof}

\section*{Supplementary Note 3 -- Pseudocode for the training algorithms}

\begin{algorithm}[H]
  \caption{\textit{Joint}}
\begin{algorithmic}[1]
  \Require The gradient penalty coefficient $\lambda$, the number of critic parameters $\vec{w}$ updates per one model iteration $N_{\vec{w}}$, the batch size $B$, Adam hyperparameters for quantum circuit, observable and critic parameters $(\gamma^{\vec{\theta}}, \beta_1^{\vec{\theta}}, \beta_2^{\vec{\theta}}, \gamma^{\vec{\vec{\alpha}}}, \beta_1^{\vec{\vec{\alpha}}}, \beta_2^{\vec{\vec{\alpha}}}, \gamma^{\vec{w}}, \beta_1^{\vec{w}}, \beta_2^{\vec{w}})$.
  \Require Initial quantum circuit, observable, and critic 
    parameters $(\vec{w}_0, \vec{\theta}_0, \vec{\alpha}_0)$.
    \While{$(\vec{\theta}, \vec{\alpha})$ has not converged}
      \For{$t = 1, \dots, N_{\vec{w}}$}
      \For{$i = 1, \dots, B$}
        \State Sample real data $\vec{x}^{(i)} \sim \mathbb{Q}$, latent variable $\vec{z}^{(i)} \sim \mathbb{P}_z$, a random number $\epsilon \sim U[0,1]$.
        \State $\tilde{\vec{x}}^{(i)} \gets G_{\vec{\theta},\vec{\alpha}} (\vec{z}^{(i)})$
        \State $\hat{\vec{x}}^{(i)} \gets \epsilon \vec{x}^{(i)} + (1-\epsilon)\tilde{\vec{x}}^{(i)}$
      \EndFor
      \State $\mathcal{L}_C \gets \frac{1}{B}\sum_{i=1}^B \left[ D_{\vec{w}}(\tilde{\vec{x}}^{(i)}) - D_{\vec{w}}(\vec{x}^{(i)})+ \lambda \left(\norm{\nabla_{\hat{\vec{x}}}D_{\vec{w}}(\hat{\vec{x}}^{(i)})}_2 -1\right)^2 \right]$
      \State $\vec{w} \gets$ Adam$(\nabla_{\vec{w}} \mathcal{L}_{C}, \vec{w}, \gamma^{\vec{w}}, \beta_1^{\vec{w}}, \beta_2^{\vec{w}})$ 
      \EndFor
      \color{blue}
      \State Sample a batch of latent variables $\{\vec{z}^{(i)}\}_{i=1}^B \sim \mathbb{P}_z$
      \State $\mathcal{L}_G \gets -\frac{1}{B}\sum_{i=1}^B D_{\vec{w}}(G_{\vec{\theta}, \vec{\alpha}}(\vec{z}^{(i)}))$
      \State $(\vec{\theta}, \vec{\alpha}) \gets $ Adam$(\nabla_{\vec{\theta}, \vec{\alpha}} \mathcal{L}_G, (\vec{\theta}, \vec{\alpha}), (\gamma^{\vec{\theta}}, \gamma^{\vec{\alpha}}), (\beta_1^{\vec{\theta}}, \beta_1^{\vec{\alpha}}), (\beta_2^{\vec{\theta}}, \beta_2^{\vec{\alpha}}))$
      \color{black}
    \EndWhile
\end{algorithmic}
\label{alg: joint}
\end{algorithm}
\vspace{-0.6cm}

\begin{algorithm}[H]
\caption{\textit{Asynchronous}}
\begin{algorithmic}[1]
    \Require The gradient penalty coefficient $\lambda$, the number of critic parameters $\vec{w}$ updates per one model iteration $N_{\vec{w}}$, the number of observable parameters $\vec{\alpha}$ updates per one model iteration $N_{\vec{\alpha}}$, the batch size $B$, Adam hyperparameters for quantum circuit, observable and critic parameters $(\gamma^{\vec{\theta}}, \beta_1^{\vec{\theta}}, \beta_2^{\vec{\theta}}, \gamma^{\vec{\vec{\alpha}}}, \beta_1^{\vec{\vec{\alpha}}}, \beta_2^{\vec{\vec{\alpha}}}, \gamma^{\vec{w}}, \beta_1^{\vec{w}}, \beta_2^{\vec{w}})$.
  \Require Initial quantum circuit, observable, and critic 
    parameters $(\vec{w}_0, \vec{\theta}_0, \vec{\alpha}_0)$.
    \While{$(\vec{\theta}, \vec{\alpha})$ has not converged}
      \For{$t = 1, \dots, N_{\vec{w}}$}
      \For{$i = 1, \dots, B$}
        \State Sample real data $\vec{x}^{(i)} \sim \mathbb{Q}$, latent variable $\vec{z}^{(i)} \sim \mathbb{P}_z$, a random number $\epsilon \sim U[0,1]$.
        \State $\tilde{\vec{x}}^{(i)} \gets G_{\vec{\theta},\vec{\alpha}} (\vec{z}^{(i)})$
        \State $\hat{\vec{x}}^{(i)} \gets \epsilon \vec{x}^{(i)} + (1-\epsilon)\tilde{\vec{x}}^{(i)}$
      \EndFor
      \State $\mathcal{L}_C \gets \frac{1}{B}\sum_{i=1}^B \left[D_{\vec{w}}(\tilde{\vec{x}}^{(i)}) - D_{\vec{w}}(\vec{x}^{(i)})+ \lambda \left(\norm{\nabla_{\hat{\vec{x}}}D_{\vec{w}}(\hat{\vec{x}}^{(i)})}_2 -1\right)^2\right]$
      \State $\vec{w} \gets$ Adam$(\nabla_{\vec{w}} \mathcal{L}_{C}, \vec{w}, \gamma^{\vec{w}}, \beta_1^{\vec{w}}, \beta_2^{\vec{w}})$ 
      \EndFor
   \color{blue}
      \State Sample a batch of latent variables $\{\vec{z}^{(i)}\}_{i=1}^B \sim \mathbb{P}_z$.
      \For {$j=1,\dots, N_{\vec{\alpha}}$}
        \State $\mathcal{L}_G \gets -\frac{1}{B}\sum_{i=1}^B D_{\vec{w}}(G_{\vec{\theta}, \vec{\alpha}}(\vec{z}^{(i)}))$
        \State $\vec{\alpha} \gets $ Adam$(\nabla_{\vec{\alpha}} \mathcal{L}_G, \vec{\alpha}, \gamma^{\vec{\alpha}}, \beta_1^{\vec{\alpha}}, \beta_2^{\vec{\alpha}})$
      \EndFor
      \State $\mathcal{L}_G \gets -\frac{1}{B}\sum_{i=1}^B D_{\vec{w}}(G_{\vec{\theta}, \vec{\alpha}}(\vec{z}^{(i)}))$
      \State $\vec{\theta} \gets $ Adam$(\nabla_{\vec{\theta}} \mathcal{L}_G, \vec{\theta}, \gamma^{\vec{\theta}}, \beta_1^{\vec{\theta}}, \beta_2^{\vec{\theta}})$ 
    \color{black}
    \EndWhile
\end{algorithmic}
\end{algorithm}
\vspace{-0.6cm}

\begin{algorithm}[H]
  \caption{\textit{Decoupled}}
\begin{algorithmic}[1]
  \Require The gradient penalty coefficient $\lambda$, the number of critic parameters $\vec{w}$ updates per iteration $N_{\vec{w}}$, the number of observable parameters $\vec{\alpha}$ updates per iteration $N_{\vec{\alpha}}$, the batch size $B$, Adam hyperparameters for quantum circuit, observable and critic parameters $(\gamma^{\vec{\theta}}, \beta_1^{\vec{\theta}}, \beta_2^{\vec{\theta}}, \gamma^{\vec{\vec{\alpha}}}, \beta_1^{\vec{\vec{\alpha}}}, \beta_2^{\vec{\vec{\alpha}}}, \gamma^{\vec{w}}, \beta_1^{\vec{w}}, \beta_2^{\vec{w}})$.
  \Require Initial quantum circuit, observable and critic parameters $(\vec{w}_0, \vec{\theta}_0, \vec{\alpha}_0)$.
     \While{$(\vec{\theta}, \vec{\alpha})$ has not converged}
    \color{blue}\For{$t = 1, \dots, N_{\vec{\alpha}}$}
      \color{black}\For{$t = 1, \dots, \lceil N_{\vec{w}} / N_{\vec{\alpha}} \rceil $}
      \For{$i = 1, \dots, B$}
        \State Sample real data $\vec{x}^{(i)} \sim \mathbb{Q}$, latent variable $\vec{z}^{(i)} \sim \mathbb{P}_z$, a random number $\epsilon \sim U[0,1]$.
        \State $\tilde{\vec{x}}^{(i)} \gets G_{\vec{\theta},\vec{\alpha}} (\vec{z}^{(i)})$
        \State $\hat{\vec{x}}^{(i)} \gets \epsilon \vec{x}^{(i)} + (1-\epsilon)\tilde{\vec{x}}^{(i)}$
      \EndFor
      \State $\mathcal{L}_C \gets \frac{1}{B}\sum_{i=1}^B \left[D_{\vec{w}}(\tilde{\vec{x}}^{(i)}) - D_{\vec{w}}(\vec{x}^{(i)})+ \lambda \left(\norm{\nabla_{\hat{\vec{x}}}D_{\vec{w}}(\hat{\vec{x}}^{(i)})}_2 -1\right)^2 \right]$
      \State $\vec{w} \gets$ Adam$(\nabla_{\vec{w}} \mathcal{L}_{C}, \vec{w}, \gamma^{\vec{w}}, \beta_1^{\vec{w}}, \beta_2^{\vec{w}})$ 
      \EndFor \color{blue} 
      \State $\mathcal{L}_G \gets -\frac{1}{B}\sum_{i=1}^B D_{\vec{w}}(G_{\vec{\theta}, \vec{\alpha}}(\vec{z}))$
      \State $\vec{\alpha} \gets $ Adam$(\nabla_{\vec{\alpha}} \mathcal{L}_G, \vec{\alpha}, \gamma^{\vec{\alpha}}, \beta_1^{}, \beta_2^{\vec{\alpha}})$ \EndFor
    
    \State Sample a batch of latent variables $\{\vec{z}^{(i)}\}_{i=1}^B \sim \mathbb{P}_z$.
      \State $\mathcal{L}_G \gets -\frac{1}{B}\sum_{i=1}^B D_{\vec{w}}(G_{\vec{\theta}, \vec{\alpha}}(\vec{z}))$
      \color{blue}\State $\vec{\theta} \gets $ Adam$(\nabla_{\vec{\theta}} \mathcal{L}_G, \vec{\theta}, \gamma^{\vec{\theta}}, \beta_1^{\vec{\theta}}, \beta_2^{\vec{\theta}})$ 
  \color{black}\EndWhile
\end{algorithmic}
\end{algorithm}

\section*{Supplementary Note 4 -- Numerical Experiment Setup}
\label{app:numerical details}
We use the Python package Tensorcircuit \cite{zhang_tensorcircuit_2023} for constructing quantum circuits, Equinox \cite{kidger_equinox_2021} for constructing the remaining architecture of the generator and the critic, \texttt{Jax} \cite{bradhury_jax_2018} for the simulation of training and sampling and \texttt{FAISS} \cite{douze_faiss_2024} for the $k$-NN subroutine in the KLD estimator (Equation 25 of Ref. \cite{wang_divergence_2009}). We perform all experiments on a single \texttt{NVIDIA RTX 2080 Ti GPU}.

This stack prioritizes the computational speed for simulating the training and sampling of OT-EVSs by making a copy of the quantum circuit for each observable, for each batch sample, and parallelizing all the circuits on the GPU. Our simulations, however, put a relatively high burden on GPU memory. The GPU memory consumption scales not only in the number of qubits $n$, the number of Pauli strings $L$, and the batch size $B$ but also significantly in the number of quantum circuit parameters $\vec{\vec{\theta}}$ because of the gradient storage in automatic differentiation. Simulations of the shadow measurement scheme, for which the expectation values of $2k-$local observables need to be computed for variance approximation (Methods), consume more memory than those of the conventional measurement scheme by an asymptotic factor of $n^k$. Consequently, we only use conventional measurements for some of our numerical experiments with larger models.

\begin{figure*}[h]
\begin{tabular}{ll}
$U_{n=4}^{B}(\vec{\vec{\theta}})$ =
\scalebox{1}{\begin{quantikz}[row sep = {0.8cm, between origins}, column sep=0.2cm]
   & \gate[1][1.15cm]{Y(\vec{\theta}^1_1)} & \ctrl{1} & &
\\
  & \gate[1][1.15cm]{Y(\vec{\theta}^1_2)} & \control{} & \ctrl{1} & 
\\  
  & \gate[1][1.15cm]{Y(\vec{\theta}^1_3)} & \ctrl{1} & \control{} &
\\  
    & \gate[1][1.15cm]{Y(\vec{\theta}^1_4)} & \control{} && 
\end{quantikz}}
\hspace{2.05cm} 
$U_{n=4}^{S}(\vec{\vec{\theta}})$ =
\scalebox{1}{\begin{quantikz}[row sep = {0.8cm, between origins}, column sep=0.2cm]
   & \gate[1][1.15cm]{Y(\vec{\theta}^1_1)} & \ctrl{1} & & \ctrl{1} & &&  & &&&
\\
  & \gate[1][1.15cm]{Y(\vec{\theta}^1_2)} & \targ{} & \gate[1][1.15cm]{Y(\vec{\theta}^1_{n+1})} & \targ{}   &  \ctrl{1} & & \ctrl{1}           & &&&
\\  
  & \gate[1][1.15cm]{Y(\vec{\theta}^1_3)} &&& & \targ{} &  \gate[1][1.15cm]{Y(\vec{\theta}^1_{n+2})}  & \targ{} & \ctrl{1}& & \ctrl{1} &
\\  
    & \gate[1][1.15cm]{Y(\vec{\theta}^1_4)} &&&&&& & \targ{} &  \gate[1][1.15cm]{Y(\vec{\theta}^1_{n+3})}  & \targ{} & 
\end{quantikz}}
\end{tabular}
\caption{The circuit diagrams for one layer of brickwork ansatz (left) and one layer of sequential ansatz (right) for $4$ qubits.}
\label{fig: circuit diagrams}
\end{figure*}

We examined two circuit ansatze, \textit{sequential} and \textit{brickwork}, each with a different encoding circuit. The circuit diagrams for one layer of sequential ansatz (denoted by $U^S$) and one layer of brickwork ansatz (denoted by $U^B$) are shown in Figure \ref{fig: circuit diagrams}. We consider different combinations of latent variable embedding strategies and variational ansatz. For the illustrative example in ``Methods'', we use $$U^S \bigotimes_{j=1}^{n}X(z_2)U^S\bigotimes_{j=1}^{n}X(z_1).$$ For the sequential and brickwork circuits in ``Results'', we use 
$$\left(U^S\right)^{N_l}\bigotimes_{j=1}^{n}Z_j(z_2)\bigotimes_{j=1}^{n}X_j(z_1) \hspace{0.5cm} \text{and} \hspace{0.5cm} \left(U^B\right)^{N_l} \bigotimes_{\substack{j=1 \\ j \text{ odd}}}^{n}X_j(z_{1}) \bigotimes_{\substack{k=1 \\ k \text{ even}}}^{n} X_k(z_{2}).$$ 

In the numerical experiments for training methods and shot noise, we examined models with both ansätze across three configurations (number of qubits $n$, number of circuit layers $N_l$, locality of observables $k$, data dimension $M$):
\begin{enumerate}[label=(C\arabic*)]
\centering
\item $n=8$, $N_l=2$, $k=1$, $M=8$, 
\item $n=8$, $N_l=9$, $k=1$, $M=8$,
\item $n=11$, $N_l=2$, $k=2$, $M=64$.
\end{enumerate}

\begin{table}[h]
\caption{List of hyperparameters used in
numerical experiments. We perform hyperparameter optimization by random grid search for each experiment. All trials in each experiment (for all WGAN variants, all measurement schemes, or all circuit depths, when applicable) share the same hyperparameters.}
\begin{tabular}{ p{1.5cm} p{1.8cm} p{0.9cm} p{0.9cm} p{0.9cm} p{0.9cm} p{0.9cm} p{0.9cm} p{0.9cm} p{0.9cm} p{0.9cm} p{0.9cm} p{0.9cm} p{0.9cm} p{0.9cm}}
\hline
\hline
\centering & & \multicolumn{13}{c}{WGAN Hyperparameters} \\
\cline{3-15}
\centering & & $\lambda$ & $N_{\vec{w}}$ & $N_{\vec{\vec{\alpha}}}$ & $B$ & $\gamma^{\vec{\vec{\theta}}}$ & $\beta_1^{\vec{\vec{\theta}}}$ & $\beta_2^{\vec{\vec{\theta}}}$ & $\gamma^{\vec{\vec{\alpha}}}$ & $\beta_1^{\vec{\vec{\alpha}}}$ & $\beta_2^{\vec{\vec{\alpha}}}$ & $\gamma^{\vec{w}}$ & $\beta_1^{\vec{w}}$ & $\beta_2^{\vec{w}}$\\
\hline
\multicolumn{2}{c}{Illustration ($n=4$)} & 0.1 & 5 & 5 & 256 & $10^{-3}$ & 0 & 0.9 & $10^{-4}$ & 0.9 & 0.9 & $10^{-4}$ & 0.5 & 0.9 \\
\hline
\centering \multirow{6}{4em}{Compare Training Settings} & Seq.(C1) & 0.1 & 5 & 5 & 256 & $10^{-3}$ & 0 & 0.99 & $10^{-4}$ & 0 & 0.9 & $10^{-4}$ & 0.9 & 0.99 \\
\centering & Seq.(C2) & 0.1 & 5 & 5 & 256 & $10^{-3}$ & 0 & 0.5 & $10^{-4}$ & 0 & 0.9 & $10^{-4}$ & 0.5 & 0.9 \\
\centering & Seq.(C3) & 0.1 & 5 & 5 & 256 & $10^{-2}$ & 0.5 & 0.5 & $10^{-4}$ & 0.5 & 0.9 & $10^{-4}$ & 0.5 & 0.9 \\
\centering & Brk.(C1) & 0.1 & 5 & 5 & 256 & $10^{-3}$ & 0 & 0.9 & $10^{-4}$ & 0 & 0.99 & $10^{-4}$ & 0 & 0.99 \\
\centering & Brk.(C2) & 0.1 & 5 & 5 & 256 & $10^{-3}$ & 0 & 0.5 & $10^{-4}$ & 0 & 0.9 & $10^{-4}$ & 0.5 & 0.9 \\
\centering & Brk.(C3) & 0.1 & 5 & 5 & 256 & $10^{-2}$ & 0.5 & 0.5 & $10^{-4}$ & 0 & 0.5 & $10^{-4}$ & 0.5 & 0.9 \\
\hline 
\centering \multirow{2}{4em}{OT- vs. OF-EVS} & Seq. & 0.1 & 5 & 5 & 256 & $10^{-3}$ & 0 & 0.5 & $10^{-4}$ & 0 & 0.9 & $10^{-4}$ & 0.5 & 0.9  \\
\centering & Brk. & 0.1 & 5 & 5 & 256 & $10^{-3}$ & 0 & 0.5 & $10^{-4}$ & 0 & 0.9 & $10^{-4}$ & 0.5 & 0.9 \\
\hline
\centering \multirow{2}{4em}{Image Datasets} & MNIST & 1 & 5 & 5 & 256 & $10^{-2}$ & 0.9 & 0.9 & $10^{-3}$ & 0 & 0.99 & $10^{-3}$ & 0.5 & 0.5   \\
\centering & Fashion & 1 & 5 & 5 & 256 & $10^{-2}$ & 0.9 & 0.9 & $10^{-3}$ & 0 & 0.99 & $10^{-3}$ & 0.5 & 0.5 \\
\hline 
\hline
\end{tabular}
\label{table:hyperparameters}
\end{table}

In all numerical experiments, the models are initialized as follows before training: The quantum circuit parameters are drawn from the uniform distribution on $[-\pi, \pi)$, and the observable and critic parameters are set according to the Kaiming Initialisation \cite{he_delving_2015}. We always use the hyperparmeters $N_{\vec{w}}=N_{\vec{\alpha}}=5$ for training. The other hyperparameters for training are searched based on a randomized grid search, with results summarized in Table \ref{table:hyperparameters}. For a fair comparison, we use exactly the same hyperparameter setting for all experiment cases, i.e., across all three training algorithms, measurement schemes, and all model configurations for the same training algorithm and measurement scheme. We evaluate the model performance (KL divergence) using $2048$ generated samples (maximum allowed by FAISS \cite{douze_faiss_2024}) and $2048$ training data.

We train autoencoders to downscale the MNIST and Fashion-MNIST datasets. The two autoencoders have the same architecture, except that the data are downscaled to $8$ and $16$ dimensions, respectively. The encoder is a multi-layer perceptron of two hidden layers ($128$ and $64$ neurons) with RELU activation function. The decoder has the reversed architecture with an additional final sigmoid function.

\newpage

\begin{figure}[h!]
\includegraphics[width=0.7\linewidth]{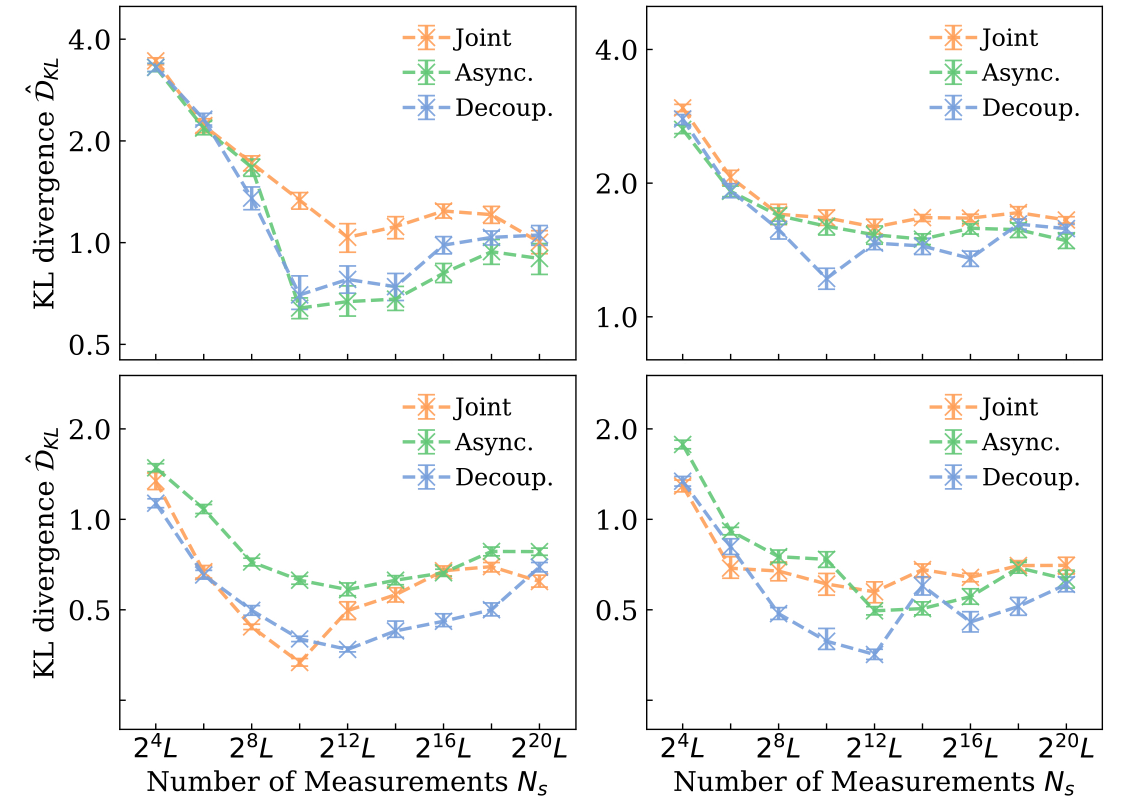}
\caption{Training performance of OT-EVS with (a) a nine-layer $8$-qubit \textit{sequential} circuit (b) a nine-layer $8$-qubit \textit{brickwork} circuit (c) a two-layer $11$-qubit \textit{sequential} circuit (d) a two-layer $11$-qubit \textit{brickwork} circuit using conventional measurements on the synthetic dataset. The interquartile mean and bootstraped $95\%$ confidence intervals over $20$ trials for the estimated KL divergence after $50k$ training iterations are shown.}
\label{fig: mnist}
\end{figure}

\begin{figure}[h!]
\includegraphics[width=0.9\linewidth]{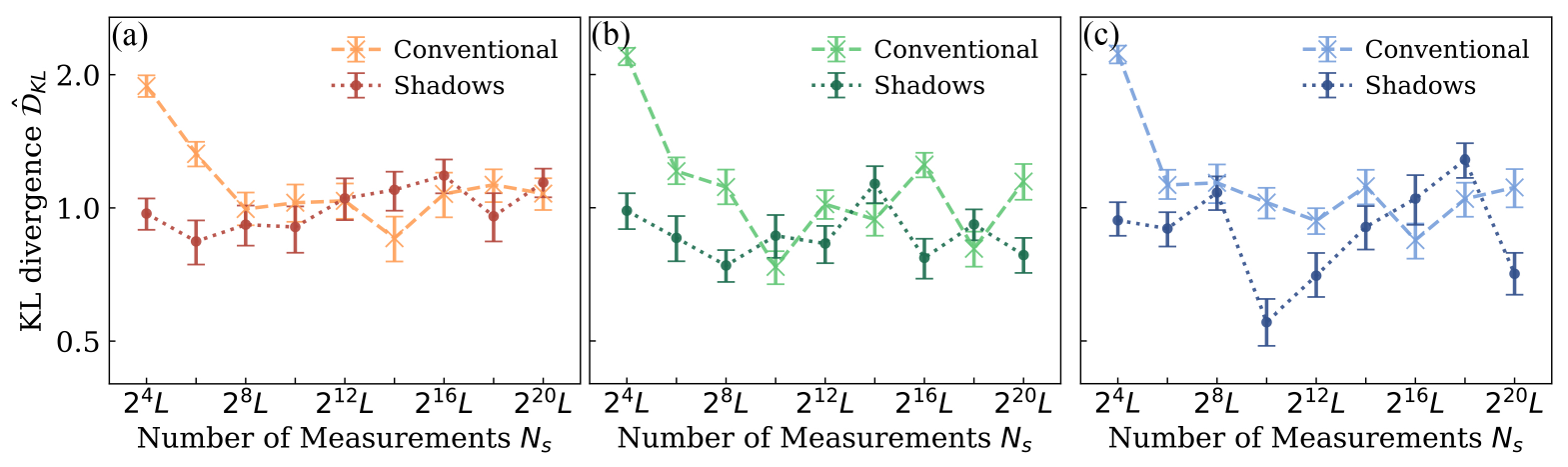}
\caption{Training performance of OT-EVS with a two-layer $8$-qubit \textit{brickwork} circuit using the (a) \textit{Joint} (b) \textit{Asynchronous} (c) \textit{Decoupled} method on the synthetic dataset. 
The interquartile mean and bootstraped $95\%$ confidence intervals over $20$ trials for the estimated KL divergence after $50k$ training iterations are shown.}
\label{fig: mnist}
\end{figure}

\end{document}